\documentclass[aps,pra,reprint,superscriptaddress,floatfix,tightenlines,nofootinbib,nobibnotes]{revtex4-2}

\usepackage{graphicx}
\usepackage{xurl}
\usepackage{booktabs}
\usepackage{comment}
\usepackage{dsfont}
\usepackage{physics}
\usepackage{nicefrac}
\usepackage{mathtools}
\usepackage{amsmath, amssymb, amsfonts, amsthm, amscd}
\usepackage{thmtools, thm-restate}
\usepackage{accents}
\usepackage{tikz}
\usepackage{bm}
\usetikzlibrary{quantikz2}

\usepackage[colorlinks,linkcolor=blue,urlcolor=blue, citecolor=blue]{hyperref}

\theoremstyle{plain}
\newtheorem{theorem}{Theorem}
\newtheorem{lemma}[theorem]{Lemma}
\newtheorem{proposition}[theorem]{Proposition}
\newtheorem{corollary}[theorem]{Corollary}

\theoremstyle{definition}

\theoremstyle{remark}
\newtheorem{remark}[theorem]{Remark}

\DeclareMathOperator{\im}{im}
\DeclareMathOperator{\wt}{wt}
\DeclareMathOperator{\sys}{Sys}
\DeclareMathOperator{\vol}{Vol}
\DeclareMathOperator{\inj}{inj}

\allowdisplaybreaks

\begin{document}

\title{Hyperbolic color codes with constant rate and polynomial distance}

\author{Shun Hasegawa}
\email{shun.hasegawa.qc@gmail.com}
\affiliation{
Department of Computer Science,
Graduate School of Information Science and Technology,
The University of Tokyo, 7–3–1 Hongo, Bunkyo-ku, Tokyo, 113–8656, Japan
}

\author{Hayata Yamasaki}
\email{hayata.yamasaki@gmail.com}
\affiliation{
Department of Computer Science,
Graduate School of Information Science and Technology,
The University of Tokyo, 7–3–1 Hongo, Bunkyo-ku, Tokyo, 113–8656, Japan
}

\begin{abstract}
Recent advances in quantum hardware relax the strict geometric-locality constraints traditionally imposed on quantum error-correcting codes, motivating interest in high-rate quantum low-density parity-check (qLDPC) codes. At the same time, color codes provide a particularly rich setting for fault-tolerant quantum computation, underlying protocols such as single-shot error correction and self-correcting quantum computation. Hyperbolic color codes provide a class of high-rate qLDPC codes that also retain the structural features of color codes relevant to fault-tolerant quantum computation. However, previous constructions of hyperbolic color codes have achieved at most logarithmic code distance. In this work, we construct hyperbolic color codes with both constant rate and polynomial distance by building on arithmetic hyperbolic manifolds that support polynomial-distance hyperbolic toric codes. Our construction applies in arbitrary dimension \(D\geq 4\). In even dimensions, the resulting type-\(D/2\) color codes have constant encoding rate and polynomial distance, while in other cases the number of logical qubits and the code distance both exhibit polynomial scaling. We further derive explicit exponents for polynomial lower bounds as functions of the dimension and code type. These results establish a family of hyperbolic color codes simultaneously achieving constant rate and polynomial distance and providing a testbed to explore fault-tolerant quantum computation protocols that combine high-rate quantum codes with the structural advantages of color codes.
\end{abstract}

\maketitle

\tableofcontents

\section{Introduction}

Reliable quantum computation requires quantum error correction, typically at the cost of substantial physical-qubit overhead~\cite{Shor:1995hbe}. Reducing this overhead has motivated growing interest in high-rate quantum low-density parity-check (qLDPC) codes~\cite{tillichQuantumLDPCCodes2009a, Breuckmann:2020jpn, panteleevAsymptoticallyGoodQuantum2022, leverrierQuantumTannerCodes2022a, dinurGoodQuantumLDPC2023, Breuckmann:2021yvk}. A nonvanishing encoding rate is necessary to achieve constant physical-qubit overhead~\cite{gottesmanFaulttolerantQuantumComputation2014, yamasakiTimeEfficientConstantSpaceOverheadFaultTolerant2024, yoshidaConcatenateCodesQubits2025, tamiyaFaulttolerantQuantumComputation2026}, while polynomial distance is required for exponential error suppression in a power of the block size when the logical error rate decreases exponentially with distance.

Recent experimental advances have demonstrated that the requirement for strictly local connectivity on a two-dimensional lattice can be relaxed, for example through dynamically reconfigurable neutral-atom arrays~\cite{bluvsteinLogicalQuantumProcessor2024} and ion-shuttling architectures~\cite{pinoDemonstrationTrappedionQuantum2021}. These developments broaden the range of quantum codes that may be practically relevant and motivate the exploration of code families with distinct structural and operational advantages.

\begin{figure*}[t]
    \centering
    \includegraphics[width=\textwidth]{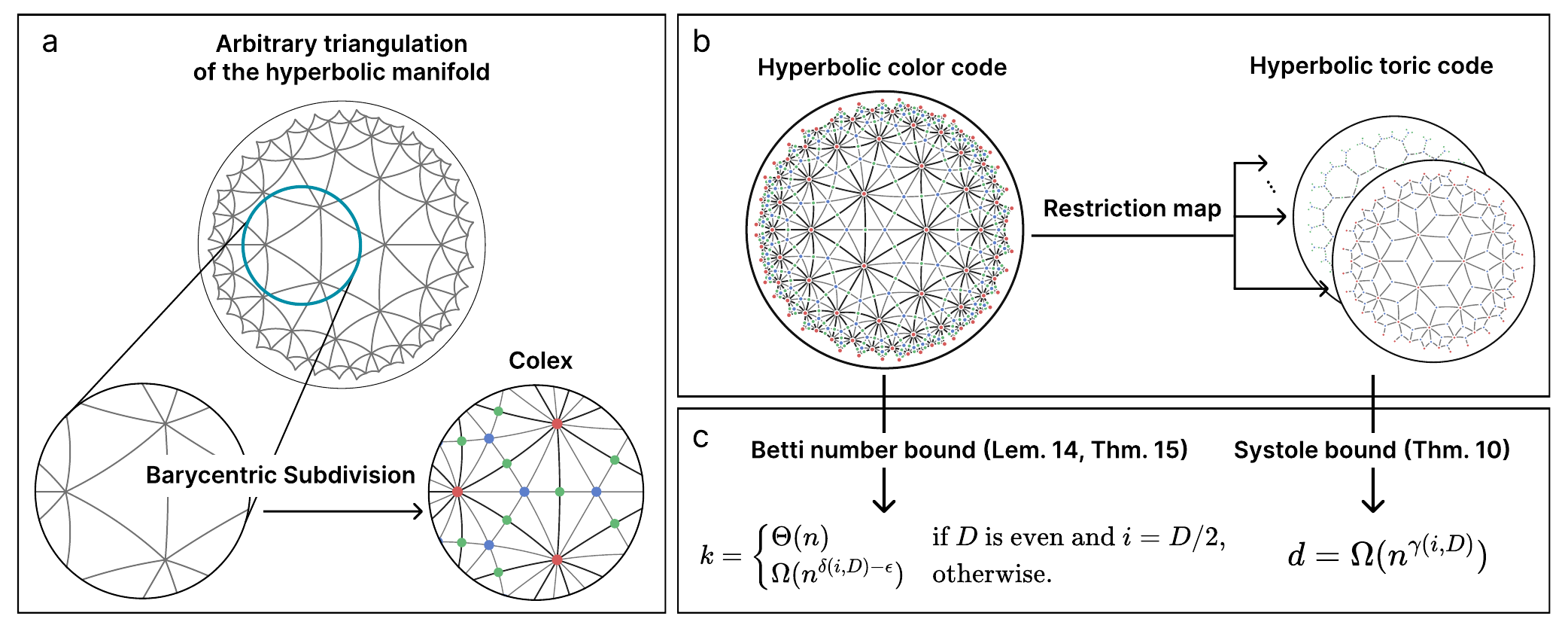}
    \caption{Overview of the construction and main results of this work.
(a) Starting from an arbitrary triangulation of a closed hyperbolic manifold, we obtain a colex by barycentric subdivision, which ensures the required vertex colorability for defining a color code (Sec.~\ref{sec:construction-of-hyperbolic-color-codes}).
(b) To derive a lower bound on the distance of the resulting hyperbolic color code, we transform the color code into a collection of associated toric codes using restriction maps and show that the color-code distance is bounded from below by the distances of these toric codes (Sec.~\ref{sec:code-distances-of-color-and-toric-codes}).
(c) We explicitly determine the scaling of the number of logical qubits \(k\) and the code distance \(d\) with the number of physical qubits \(n\), where \(D\) denotes the dimension of the underlying hyperbolic space and \(i\) the type index of the color code. The Betti-number bounds determine the scaling of \(k\), while the systolic bounds, together with the result in (b), yield a polynomial lower bound on \(d\). The exponents \(\delta(i,D)\) and \(\gamma(i,D)\) in these polynomial lower bounds depend only on the code type \(i\) and the dimension \(D\) (Sec.~\ref{sec:explicit-lower-bounds}).}
    \label{fig:placeholder}
\end{figure*}

Topological quantum codes are attractive for fault-tolerant quantum computation because they combine geometrically local stabilizer measurements with mathematical structures that facilitate the analysis of code properties and logical operations. The toric code is a canonical example, defined from the chain-complex structure associated with a cellulation of a manifold~\cite{Kitaev:1997wr, kitaevQuantumComputationsAlgorithms1997}. Another important family consists of color codes~\cite{Bombin:2006sj, Bombin:2006rk}. Compared with toric codes, the color-code family includes constructions supporting a richer set of transversal logical gates~\cite{Bombin:2015tpp, Kubica:2014jue} and has served as a useful testbed for fault-tolerant protocols, including self-correcting quantum computation~\cite{Bombin:2012cnq} and single-shot error correction~\cite{bombinSingleShotFaultTolerantQuantum2015}.

Hyperbolic manifolds provide geometries on which topological codes can simultaneously achieve a constant encoding rate and a code distance that grows with the number of physical qubits~\cite{zemorCayleyGraphsSurface2009, kimQuantumCodesHurwitz2007, Delfosse:2013smk}, a combination that is impossible for codes built on Euclidean lattices~\cite{Bravyi:2009ert}. In four dimension, toric codes defined on suitable hyperbolic manifolds achieve constant rate together with polynomial distance~\cite{Guth:2014zxf}, and explicit constructions of such hyperbolic toric codes, including concrete families of manifolds and associated decoders, are known~\cite{Londe:2019jie, Breuckmann:2020jyn}.

For color codes, comparable constructions have been reported in two and three dimensions. In two dimensions, hyperbolic color codes with constant rate and distance scaling logarithmically in the number of physical qubits are known~\cite{Delfosse:2013smk}. A different family of two-dimensional hyperbolic color codes achieves an encoding rate approaching one as the number of physical qubits tends to infinity, but at the cost of a code distance fixed at four~\cite{Soares:2018ynz}. In three dimensions, a family of color codes on Torelli mapping-torus manifolds achieves constant rate, but the scaling of its code distance with \(n\) has not been established~\cite{Zhu:2023xfg}. However, none of these constructions achieves constant rate with polynomial distance, and hyperbolic color codes with these properties in four or higher dimensions have not been systematically explored.

In this work, we construct hyperbolic color codes in every dimension \(D\geq4\). In even \(D\), the type-\(D/2\) color codes achieve a constant encoding rate, while in odd dimensions, the number of logical qubits still grows polynomially in the number of physical qubits \(n\). In both cases, the code distance of the type-\(i\) color code grows polynomially in \(n\), with an exponent depending on \(i\). These results rest on three ingredients: a construction of color codes on an arbitrary triangulated manifold via barycentric subdivision, a decomposition into toric codes that yields a lower bound on the color-code distance, and an explicit determination of the resulting exponents for polynomial lower bounds for every dimension \(D\) and type \(i\).

\begin{enumerate}
    \item First, we construct hyperbolic color codes on the family of arithmetic hyperbolic manifolds \(M_N\) introduced in Ref.~\cite{Guth:2014zxf}, whose associated toric codes are known to achieve constant rate and polynomial distance. To define a color code on a triangulation of \(M_N\), one needs to construct a colex~\cite{Bombin:2006rk, Kubica:2019qji}, for which the general cell-inflation construction of Ref.~\cite{Bombin:2006rk} provides one possible approach. However, when applied to a triangulation, this construction requires passing to the dual cell complex, for which it is not immediate that the recursive colex definition of Ref.~\cite{Kubica:2019qji} is satisfied. Therefore, we instead use barycentric subdivision, which can be viewed as the dual-lattice counterpart of the cell-inflation construction~\cite[Appendix B]{Bombin:2006rk}, and prove that the barycentric subdivision of any triangulation of a closed \(D\)-manifold yields a colex (Proposition~\ref{prop:barycentric_subdivision_derive_colex} in Sec.~\ref{sec:barycentric_subdivision}).
    \item Second, we establish a polynomial lower bound on the code distance. Directly bounding the distance of a color code is difficult, so we instead relate it to the distance of an associated toric code using the restriction map of Refs.~\cite{Kubica:2015mta,Kubica:2019qji}. Although this color-to-toric transformation is known, an explicit proof that it yields a lower bound on the color-code distance has not been provided. Therefore we rigorously prove that the distance of the color code is bounded below by the distances of the toric codes obtained through the restriction maps (Secs.~\ref{sec:color-to-toric-morphism} and~\ref{sec:color-code-lower-bound}). A second difficulty is that the resulting toric code is defined on a restricted lattice rather than directly on the triangulation of \(M_N\), and hence the distance bound of Ref.~\cite{Guth:2014zxf} does not apply directly. To address this issue, we extend the argument of Ref.~\cite{Guth:2014zxf}, using the fact that the relevant nontrivial cycles of the restricted lattice are inherited from the original color-code complex, and derive a systolic lower bound that grows polynomially with the number of physical qubits (Sec.~\ref{sec:distance_lower_bound_of_hyperbolic_toric_code}). Combining these two results establishes polynomial distance for the hyperbolic color codes.
    \item Third, specializing to the manifolds \(M_N\), we derive explicit scaling bounds on both the number of logical qubits and the code distance (Sec.~\ref{sec:explicit-lower-bounds}). To bound the number of logical qubits, we need to estimate the corresponding Betti numbers. While Ref.~\cite{Guth:2014zxf} established linear growth of the second Betti number in four dimensions, its scaling in other dimensions was not addressed, and hence the scaling of the number of logical qubits remained unknown for general dimension \(D\) and type \(i\). Applying the result of Ref.~\cite{xueBettiNumbersHyperbolic1992}, we show that, for \(D\geq 4\) and \(2\leq i\leq D-2\), the \(i\)-th Betti number grows polynomially with the number of physical qubits. Since the numbers of logical qubits of both toric and color codes are determined by the corresponding Betti number, this establishes polynomial growth of the number of logical qubits with \(n\) for both code families in general dimensions. Moreover, for even \(D\), we independently apply the Lück approximation theorem of Ref.~\cite{abertGrowth$L^2$invariantsSequences2017a} to show that the \(D/2\)-th Betti number grows linearly (Appendix~\ref{app:proof-of-constant-rate}). This implies that the type-\(D/2\) hyperbolic color codes have constant rate (Theorem~\ref{thm:number-of-logical-qubits-color-codes}), and likewise that the corresponding type-\(D/2\) hyperbolic toric codes have constant rate (Theorem~\ref{thm:num-of-logical-toric-code}). To bound the code distance, we similarly need explicit estimates of the corresponding systoles. Previous work established polynomial systolic growth, but Ref.~\cite{Guth:2014zxf} did not provide explicit exponents for polynomial lower bounds, while Ref.~\cite{Londe:2019jie} gave an explicit analysis only for a specific four-dimensional construction with a fixed tessellation. To obtain explicit bounds in a more general setting, we apply the systolic bound of Ref.~\cite{kimSystoleLocallySymmetric2020}, extend it from prime to arbitrary ideals (Appendix~\ref{app:composite-ideal-systole}), and explicitly determine the distance exponent for every \(D\geq4\) and \(2\leq i\leq D-2\).
\end{enumerate}

Overall, we prove the following.

\begin{theorem}[Constant-rate polynomial-distance color codes in even dimensions]
    There exists a sequence of triangulated \(2m\)-dimensional hyperbolic manifolds such that the associated type-\(m\) color codes have parameters
    \begin{align}
        [[n,\ \Theta(n),\ \Omega\bigl(n^{\frac{m-1}{m(2m+1)}}\bigr)]].
    \end{align}
\end{theorem}

Apart from this constant-rate case, the number of logical qubits still grows at least polynomially in \(n\). For these cases as well, we compute the corresponding scaling exponent explicitly as a function of \(i\) and \(D\), for both the code distance and the number of logical qubits. This yields a complete table of the exponents governing \(d\) and \(k\) for the type-\(i\) color code, for every dimension \(D\ge4\) and every \(2\le i\le D-2\) (Table~\ref{tab:systole-exponents} and Table~\ref{tab:xue-exponents}, respectively).

Taken together, these results establish hyperbolic color codes as a viable family of constant-rate, polynomial-distance quantum LDPC codes. The existence of such codes may enable new constant-space-overhead fault-tolerant quantum computation protocols, complementing the existing high-rate qLDPC constructions. The fact that our construction yields color codes, rather than generic CSS codes, makes it a valuable testbed for exploring fault-tolerant protocols that are naturally studied in the color code setting, including the transversal implementation of logical gates, single-shot error correction, and self-correcting behavior.

\section{Preliminaries}

In this section, we introduce the basic notation used to describe topological codes. We summarize CSS codes in terms of algebraic topology in Sec.~\ref{sec:CSS_codes_and_chain_complexes}, where the number of logical qubits is given by a Betti number and the code distance by the minimum weight of a nontrivial cycle. We then define simplices and cells in Sec.~\ref{sec:simplices_and_cells} and the generalized boundary operator in Sec.~\ref{sec:generalized_boundary_operator}. Using this notation, we introduce the toric and color codes in Sec.~\ref{sec:toric_codes} and Sec.~\ref{sec:color_codes}, respectively.

\subsection{CSS codes and chain complexes}
\label{sec:CSS_codes_and_chain_complexes}

CSS codes are a family of stabilizer codes whose stabilizer group is generated by \(X\)-type and \(Z\)-type stabilizer generators. Such codes can be specified in terms of a chain complex~\cite{Bombin:2006cd},
\begin{align}
    \mathcal{C} \coloneqq  \left(
    \begin{CD}
        C_Z @>{\partial_Z}>> C_Q @>{\partial_X}>> C_X
    \end{CD}
    \right),
\end{align}
where \(C_Z\), \(C_Q\), and \(C_X\) are \(\mathbb{F}_2\)-vector spaces spanned by the \(Z\)-type stabilizer generators, the physical qubits, and the \(X\)-type stabilizer generators, respectively. The parity-check matrices \(H_Z^T\) and \(H_X\) of the CSS code are defined as the \(\mathbb{F}_2\)-matrix representations of \(\partial_Z\) and \(\partial_X\) with respect to these bases. With this identification, the chain-complex condition \(\partial_X \circ \partial_Z = 0\) translates directly into \(H_X H_Z^T = 0\) over \(\mathbb{F}_2\), which is precisely the condition that the \(X\)-type and \(Z\)-type stabilizer generators mutually commute.

Following the terminology of algebraic topology, we refer to an element \(c \in C_Q\) as a \emph{chain}, and for each chain \(c\), we write \(Z(c) \coloneqq  \prod_{i \,:\, c_i = 1} Z_i\) for the corresponding \(Z\)-type Pauli operator supported on \(c\); that is, \(Z(c)\) acts as \(Z\) on qubit \(i\) whenever \(c_i = 1\) and as the identity otherwise. With this notation, each basis element \(e_j \in C_Z\) corresponds to a \(Z\)-type stabilizer generator \(Z(\partial_Z e_j)\), whose support on the physical qubits is given by the chain \(\partial_Z e_j \in C_Q\).

More generally, for an arbitrary chain \(c \in C_Q\), the image \(\partial_X(c) \in C_X\) is precisely the \(X\)-type syndrome associated with the Pauli error \(Z(c)\); thus \(c \in \ker \partial_X\) if and only if \(Z(c)\) commutes with every \(X\)-stabilizer generator, i.e., yields a trivial syndrome.

Within \(\ker \partial_X\), the trivial chains \(\operatorname{im} \partial_Z\) are exactly the supports of \(Z\)-stabilizers, so quotienting by them identifies chains that act identically on the code space. This quotient defines the first homology group of the complex,
\begin{align}
    H_1(\mathcal{C}) \coloneqq  \ker \partial_X / \operatorname{im} \partial_Z.
\end{align}
A chain \(c\) in the trivial class \([c] = 0\) satisfies \(c \in \operatorname{im} \partial_Z\), so \(Z(c)\) is itself a \(Z\)-stabilizer and acts as the identity on the code space, whereas a chain representing a nontrivial class \([c] \neq 0\) corresponds to an operator \(Z(c)\) that commutes with all stabilizers without being one, and hence acts nontrivially on the code space. Choosing \(k\) chains \(\{c_i\}_{i=1}^{k}\) representing a basis of \(H_1(\mathcal{C})\) thus yields a logical Pauli \(Z\) basis, \(\{\bar{Z}_i \coloneqq  Z(c_i)\}_{i=1}^{k}\), for the code.

The same argument applies directly to \(X\)-type operators upon passing to the \emph{dual chain complex}
\begin{align}
    \mathcal{C}^* \coloneqq  \left(
    \begin{CD}
        C_X @>{\partial_X^T}>> C_Q @>{\partial_Z^T}>> C_Z
    \end{CD}
    \right),
\end{align}
obtained by reversing the arrows of \(\mathcal{C}\) and taking the transposes \(\partial_X^T, \partial_Z^T\) of the original boundary maps, which again satisfy \(\partial_Z^T \circ \partial_X^T = 0\). We write \(X(c) \coloneqq  \prod_{i \,:\, c_i = 1} X_i\) for the \(X\)-type Pauli operator supported on a chain \(c \in C_Q\). Each basis element \(f_j \in C_X\) corresponds to an \(X\)-type stabilizer generator \(X(\partial_X^T f_j)\), and \(\partial_Z^T(c)\) gives the \(Z\)-type syndrome of \(X(c)\), so that \(c \in \ker \partial_Z^T\) if and only if \(X(c)\) commutes with every \(Z\)-stabilizer generator. The logical Pauli \(X\) basis is then obtained by choosing \(k\) chains representing a basis of independent nontrivial classes in \(H_1(\mathcal{C^*}) \coloneqq  \ker \partial_Z^T / \operatorname{im} \partial_X^T\), giving \(\{\bar{X}_i \coloneqq  X(c_i)\}_{i=1}^{k}\).

The number of logical qubits \(k\) is then given by the rank of the first homology group of the complex, i.e., its first Betti number,
\begin{align}
    k \coloneqq  \dim H_1(\mathcal{C}) = \dim H_1(\mathcal C^*).
\end{align}

For a homology class \([c] \in H_1(\mathcal{C})\), represented by a chain \(c \in \ker \partial_X\), we define its weight as the minimum weight over all representatives in its coset,
\begin{align}
    \wt([c]) \coloneqq  \min_{c' \,\in\, c + \operatorname{im} \partial_Z} \wt(c'),
\end{align}
where, for a chain \(c' = \sum_i c'_i\, e_i \in C_Q\) expressed in the standard basis \(\{e_i\}\) of physical qubits, \(\wt(c')\) denotes its \emph{Hamming weight},
\begin{align}
    \wt(c') \coloneqq  |\{\, i \mid c'_i \neq 0 \,\}|,
\end{align}
i.e., the number of qubits on which \(Z(c')\) acts nontrivially.
The code distance \(d\) is then defined as the minimum weight of a nontrivial homology class in either sector,
\begin{align}
    d \coloneqq  \min\left( \min_{[c] \in H_1(\mathcal{C}) \setminus \{0\}} \wt([c]), \ \min_{[c] \in H_1(\mathcal{C^*}) \setminus \{0\}} \wt([c]) \right).
\end{align}

Equivalently, since every nontrivial homology class is represented by a chain in \(\ker \partial_X \setminus \operatorname{im} \partial_Z\), and every such chain represents a nontrivial class, the distance can be written directly in terms of the weight of chains as
\begin{align}
    d = \min\left( \min_{c \,\in\, \ker \partial_X \setminus \operatorname{im} \partial_Z} \wt(c), \ \min_{c \,\in\, \ker \partial_Z^T \setminus \operatorname{im} \partial_X^T} \wt(c) \right).
\end{align}

\subsection{Simplices and cells}
\label{sec:simplices_and_cells}

To construct a CSS code from a chain complex, one must specify the \(\mathbb{F}_2\)-vector spaces \(C_Z, C_Q\), and \(C_X\) together with the boundary operators \(\partial_Z\) and \(\partial_X\). For topological codes in particular, these spaces and operators arise from concrete topological objects: \(C_Z, C_Q\), and \(C_X\) are spanned by simplices or cells of a given dimension, and the boundary operators are induced by the geometric boundary map. We therefore begin by introducing cells and cell complexes, then specialize to simplices, simplicial complexes, and triangulations, before using these objects to construct the chain complex \(\mathcal{C}\) in the following sections.

For \(i\ge 0\), an \emph{\(i\)-cell} \(\sigma\) is a subset of a topological space homeomorphic to the open \(i\)-ball.

A \emph{cell complex} \(\mathcal L\) is a topological space decomposed into finitely many pairwise disjoint cells,
\begin{align}
    \mathcal L = \bigsqcup_{\sigma\in \Lambda(\mathcal L)} \sigma,
\end{align}
such that, for every \(i\)-cell \(\sigma\in\Lambda(\mathcal{L})\),
\begin{enumerate}
    \item its closure \(\bar{\sigma}\) is homeomorphic to the closed \(i\)-ball, with \(\sigma\) corresponding to its interior under this homeomorphism, and
    \item its boundary \(\bar{\sigma}\setminus \sigma\) is a union of cells of dimension strictly less than \(i\).
\end{enumerate}
We write \(\Lambda_i(\mathcal L)\) for the set of \(i\)-cells of \(\mathcal L\), so that \(\Lambda(\mathcal L)= \bigcup_{i=0}^{D}\Lambda_i(\mathcal L)\). A cell \(\tau\) is called a \emph{face} of \(\sigma\) if \(\tau\subseteq \bar{\sigma}\). The \emph{dimension} of \(\mathcal L\), denoted \(\dim \mathcal L\), is the largest \(i\) such that \(\Lambda_i(\mathcal L)\neq \emptyset\); equivalently, \(\mathcal L\) is said to be \(D\)-dimensional if \(D=\dim\mathcal L\). Throughout this paper, all cell complexes are finite and regular in the usual CW-complex terminology; see Ref.~\cite[Appendix]{hatcherAlgebraicTopology2002} for details.

Given a cell complex \(\mathcal L\), we define the \emph{\(i\)-th chain space} \(C_i(\mathcal L)\) as the \(\mathbb{F}_2\)-vector space spanned by the \(i\)-cells,
\begin{align}
    C_i(\mathcal L) \coloneqq  \left\{ \sum_{\sigma \in \Lambda_i(\mathcal L)} c_\sigma \, \sigma \ \middle|\ c_\sigma \in \mathbb{F}_2 \right\},
\end{align}
with \(\Lambda_i(\mathcal L)\) as its basis; an element \(c \in C_i(\mathcal L)\) is called an \emph{\(i\)-chain}.

Fix an ambient Euclidean space \(\mathbb{R}^D\). For \(0\le i\le D\), an \emph{\(i\)-simplex} \(\Delta^i\) is an \(i\)-dimensional polytope
\begin{align}
    \Delta^i= \left\{\sum_{j=0}^i t_j v_j \; \middle|\; 0 < t_j \land \sum_{j=0}^i t_j=1\right\} \subset \mathbb{R}^D,
\end{align}
spanned by an affinely independent set of vertices \(\{v_0, \ldots, v_i\} \subset \mathbb{R}^D\). For example, a 0-simplex, a 1-simplex, and a 2-simplex are a vertex, an edge, and a face, respectively. A face of a simplex, in the sense defined above for a general cell, corresponds to a subset of its vertex set: the \(k\)-faces of \(\Delta^i\) are exactly the \(k\)-simplices spanned by the \(\binom{i+1}{k+1}\) subsets of \(\{v_0,\ldots,v_i\}\) of size \(k+1\).

A cell complex \(X\) is called a \emph{simplicial complex} if every cell of \(X\) is a simplex, every simplex of \(X\) together with all of its faces belongs to some \(D\)-simplex of \(X\) (where \(D=\dim X\)), and, for any two simplices \(\sigma,\tau\in\Lambda(X)\), the intersection \(\bar\sigma\cap\bar\tau\) is either empty or the closure of a common face of \(\sigma\) and \(\tau\). Throughout this paper, we reserve the symbol \(X\) for a simplicial complex, writing \(\Delta_i(X)\) for the set of \(i\)-faces of \(X\). Unlike for a general cell complex, the faces of a simplex are determined combinatorially by its vertex set, as noted above; consequently, a simplicial complex \(X\) is completely determined by specifying only its set of top-dimensional simplices \(\Delta_D(X)\), since \(\Delta_i(X)\) for \(i<D\) is then recovered as the set of \(i\)-dimensional faces of the simplices in \(\Delta_D(X)\). In what follows, we therefore specify a simplicial complex by giving \(\Delta_D(X)\) alone.

A \emph{triangulation} of a closed \(D\)-manifold \(M\) is a \(D\)-dimensional simplicial complex \(X\) that is homeomorphic to \(M\) as a topological space. Since \(M\) is closed, every \((D-1)\)-simplex of \(X\) is a face of exactly two \(D\)-simplices. More generally, a \emph{cellulation} of \(M\) is a \(D\)-dimensional cell complex \(\mathcal{L}\) that is homeomorphic to \(M\) as a topological space; as with a triangulation, every \((D-1)\)-cell of \(\mathcal{L}\) is a face of exactly two \(D\)-cells.

Given two simplices \(\sigma\) and \(\tau\) spanned by disjoint vertex sets \(\{v_0,\ldots,v_i\}\) and \(\{w_0,\ldots,w_j\}\), respectively, such that their union \(\{v_0,\ldots,v_i\}\cup\{w_0,\ldots,w_j\}\) is affinely independent, we define their \emph{join} \(\sigma \star \tau\) as the \((i+j+1)\)-simplex spanned by the union of these vertex sets, i.e., the polytope
\begin{align}
    \sigma \star \tau = \left\{ \alpha x + \beta y \ \middle| \ x\in\sigma,\ y\in\tau,\ 0<\alpha,\beta \land \alpha+\beta=1 \right\}.
\end{align}
In particular, for a vertex \(v\) (viewed as a \(0\)-simplex) and an \(i\)-simplex \(\tau\), \(v\star\tau\) denotes the \((i+1)\)-simplex obtained by taking the convex hull of \(v\) and every point of \(\tau\), provided \(v\) together with the vertices of \(\tau\) is affinely independent. For a vertex \(v\) and a collection \(\mathcal{T}\) of simplices not containing \(v\) as a vertex, we write \(v\star\mathcal{T} \coloneqq  \{v\star\tau \mid \tau\in\mathcal{T}\}\).

\subsection{Generalized boundary operator}
\label{sec:generalized_boundary_operator}

To define toric codes and the color codes, we introduce boundary operators, which are linear maps between the chain spaces.

For a cell complex \(\mathcal{L}\), we introduce two complementary notions that record how cells of different dimensions are incident to one another. For \(\sigma \in \Lambda_i(\mathcal{L})\) and \(j \leq i\), we write
\begin{align}
    \Lambda_j(\sigma) \coloneqq  \{\tau \in \Lambda_j(\mathcal{L}) \mid \tau \subseteq \bar\sigma\}
\end{align}
for the set of \(j\)-cells contained in the closure of \(\sigma\), i.e., the \(j\)-dimensional faces of \(\sigma\). Conversely, for \(\sigma \in \Lambda_i(\mathcal{L})\) and \(j \geq i\), we define the \emph{star} of \(\sigma\) in dimension \(j\) as
\begin{align}
    \mathrm{St}_j(\sigma) \coloneqq  \{\tau \in \Lambda_j(\mathcal{L}) \mid \sigma \in \Lambda_i(\tau)\},
\end{align}
which is the set of \(j\)-cells having \(\sigma\) as one of their \(i\)-dimensional faces.

Using these two notions, we introduce a generalized boundary operator following Ref.~\cite{Kubica:2019qji}. For \(i \neq j\), we define the \(\mathbb{F}_2\)-linear map \(\partial_{i,j} \colon C_i(\mathcal{L}) \to C_j(\mathcal{L})\) by specifying its action on each basis element \(\sigma \in \Lambda_i(\mathcal{L})\) as
\begin{align}
    \partial_{i,j}(\sigma) = \begin{cases}
        \displaystyle\sum_{\tau \in \Lambda_j(\sigma)} \tau & \text{if } i > j, \\[1em]
        \displaystyle\sum_{\tau \in \mathrm{St}_j(\sigma)} \tau & \text{if } i < j,
    \end{cases}
\end{align}
so that \(\partial_{i,j}\) sends a cell \(\sigma\) to the sum of its \(j\)-dimensional faces when \(j < i\), and to the sum of the \(j\)-cells having \(\sigma\) as a face when \(j > i\).

These two constructions are related by transposition: \(\partial_{i,j}\) and \(\partial_{j,i}\) are transposes of one another as \(\mathbb{F}_2\)-matrices, i.e., 
\begin{align}
\partial_{i,j} = \partial_{j,i}^T.
\end{align}
Indeed, for \(\sigma \in \Lambda_i(\mathcal{L})\) and \(\tau \in \Lambda_j(\mathcal{L})\) with \(i>j\), the coefficient of \(\tau\) in \(\partial_{i,j}(\sigma)\) is \(1\) if and only if \(\tau\) is a face of \(\sigma\), i.e., \(\tau\subseteq\bar\sigma\). By definition, this is precisely the condition for \(\sigma\) to belong to \(\mathrm{St}_i(\tau)\), so the coefficient of \(\sigma\) in \(\partial_{j,i}(\tau)\) is also \(1\). Since the two coefficients agree for every pair \((\sigma,\tau)\), the matrix of \(\partial_{i,j}\) in the bases \(\Lambda_i(\mathcal{L})\) and \(\Lambda_j(\mathcal{L})\) is the transpose of the matrix of \(\partial_{j,i}\).

\subsection{Toric codes}
\label{sec:toric_codes}

Using the generalized boundary operators \(\partial_{i,j}\), we define \emph{toric codes} on general cell complexes~\cite{Kitaev:1997wr}. Let \(\mathcal{L}\) be a \(D\)-dimensional cell complex. For each type index \(i\) with \(1\le i \le D-1\), the pair \((\mathcal{L}, i)\) determines a CSS code, which we call the \emph{type-\(i\) toric code} on \(\mathcal{L}\), obtained from the chain complex
\begin{align}
    TC_i(\mathcal L) \coloneqq  \left(
    \begin{CD}
        C_{i+1}(\mathcal L) @>{\partial_{i+1, i}}>> C_{i}(\mathcal L) @>{\partial_{i,i-1}}>> C_{i-1}(\mathcal L)   
    \end{CD}
    \right).
\end{align}
Physical qubits are placed on the \(i\)-cells, so that \(C_i(\mathcal{L})\) plays the role of the physical qubit space, while \(C_{i-1}(\mathcal{L})\) and \(C_{i+1}(\mathcal{L})\) index the \(X\)- and \(Z\)-type stabilizer generators, respectively. The stabilizer generators are obtained directly from the boundary operators: for every \((i-1)\)-cell \(\sigma\) and \((i+1)\)-cell \(\tau\), we define
\begin{align}
    S_X(\sigma) = X(\partial_{i-1,i}(\sigma)), \quad S_Z(\tau) = Z(\partial_{i+1, i}(\tau)),
\end{align}
so that \(S_X(\sigma)\) acts on the \(i\)-cells forming the coboundary of \(\sigma\), and \(S_Z(\tau)\) acts on the \(i\)-cells forming the boundary of \(\tau\).

Logical \(Z\) operators of the type-\(i\) toric code are given by representatives of nontrivial classes in the homology group \(H_1(TC_i(\mathcal{L})) = \ker \partial_{i,i-1} / \operatorname{im} \partial_{i+1,i}\), and the number of logical qubits is \(k = \dim H_1(TC_i(\mathcal{L}))\). Logical \(X\) operators arise analogously from the homology of the transposed complex, \(H_1(TC_i^*(\mathcal{L})) = \ker \partial_{i,i+1} / \operatorname{im} \partial_{i-1,i}\), and \(k = \dim H_1(TC_i^*(\mathcal{L}))\) as well, consistent with the \(\partial_{i,j} = \partial_{j,i}^T\) relation established above.

More generally, the cell complex \(\mathcal{L}\) gives rise to a full chain complex \(C_D(\mathcal{L}) \xrightarrow{\partial_{D,D-1}} \cdots \xrightarrow{\partial_{i+1,i}} C_i(\mathcal{L}) \xrightarrow{\partial_{i,i-1}} \cdots \xrightarrow{\partial_{1,0}} C_0(\mathcal{L})\), whose \(i\)-th homology group with \(\mathbb{F}_2\) coefficients is
\begin{align}
    H_i(\mathcal L;\mathbb{F}_2) \coloneqq  \ker \partial_{i,i-1} / \operatorname{im} \partial_{i+1,i}.
\end{align}
By definition, \(H_i(\mathcal{L};\mathbb F_2) = H_1(TC_i(\mathcal{L}))\). We refer to 
\begin{align}
b_i(\mathcal{L}) \coloneqq  \dim H_i(\mathcal{L};\mathbb{F}_2).
\end{align}
as the \emph{\(i\)-th Betti number} of \(\mathcal{L}\). In particular, the number of logical qubits of the type-\(i\) toric code defined on \(\mathcal L\) is given by the \(i\)-th Betti number, i.e., 
\begin{align}
\label{eq:num_of_logical_qubits_toric_code}
    k_{TC_i(\mathcal L)}=b_i(\mathcal L).
\end{align}

The code distance is \(d = \min(d^X, d^Z)\), where \(d^Z\) is the minimum weight of a nontrivial class in \(H_1(TC_i(\mathcal{L}))\) and \(d^X\) is the minimum weight of a nontrivial class \(H_1(TC_i^*(\mathcal L))\).

\subsection{Color codes}
\label{sec:color_codes}

The \(D\)-dimensional color code is defined on a \(D\)-dimensional simplicial complex \(X\) without boundary~\cite{Bombin:2006rk}. In addition, we require that the vertices of \(X\) be \((D+1)\)-colorable, i.e., that there exist a function
\begin{align}
    \mathrm{color}\colon \Delta_0(X) \to \mathbb{Z}_{D+1} = \{0,1,\ldots,D\}.
\end{align}
such that any two vertices connected by an edge have different colors. This colorability condition is sufficient to guarantee \(\partial_{D,i-1}\circ\partial_{D-i-1,D}=0\), which is what makes it possible to define a chain complex out of these boundary operators~\cite{Kubica:2019qji}.

For each type index \(i\) with \(1\le i\le D-1\), the \emph{type-\(i\) color code} on \(X\) is obtained from the chain complex
\begin{align}
\label{eq:color-code-chain-complex}
    & CC_i(X) \coloneqq  \notag\\
    & \left(
    \begin{CD}
        C_{D-i-1}(X) @>{\partial_{D-i-1, D}}>> C_D(X) @>{\partial_{D,i-1}}>> C_{i-1}(X)
    \end{CD}
    \right).
\end{align}
Physical qubits are placed on the \(D\)-simplices, and \(X\)- and \(Z\)-type stabilizer generators are placed on the \((i-1)\)-simplices and \((D-i-1)\)-simplices, respectively. As with the toric code, the stabilizer generators are obtained directly from the boundary operators: for every \((i-1)\)-simplex \(\sigma\) and \((D-i-1)\)-simplex \(\tau\), we define
\begin{align}
    S_X(\sigma) = X(\partial_{i-1,D}(\sigma)), \quad S_Z(\tau) = Z(\partial_{D-i-1,D}(\tau)).
\end{align}

The number of logical qubits of the type-\(i\) color code on a colex \(X\) is known from Ref.~\cite{Bombin:2006rk} to be given in terms of the \(i\)-th Betti number \(b_i(X)\),
\begin{align}
\label{eq:num_of_logical_qubits_color_code}
    k_{CC_i(X)} = \binom{D}{i} b_i(X).
\end{align}
If the colex \(X\) and a cell complex \(\mathcal L\) are respectively a triangulation and a cellulation of the same manifold, then they have the same \(i\)-th Betti number. Since the number of logical qubits of the corresponding type-\(i\) toric code is \(k_{TC_i(\mathcal L)}=b_i(\mathcal L)\), it follows that
\begin{align}
\label{eq:logical-qubits-color-codes-and-toric-codes}
    k_{CC_i(X)} = \binom{D}{i} k_{TC_i(\mathcal L)}.
\end{align}

The code distance requires a separate analysis, as directly determining or bounding the distance of a color code is generally difficult. Establishing such a bound is one of the central subjects of this paper and is discussed in detail in Sec.~\ref{sec:color-code-lower-bound}.

\section{Construction of Hyperbolic Color Codes}
\label{sec:construction-of-hyperbolic-color-codes}

In this section, we propose a construction of hyperbolic color codes. First, we introduce a colex, the geometric structure used to define color codes in Sec.~\ref{sec:colexes}, following Ref.~\cite{Kubica:2019qji}. Then, we show that the barycentric subdivision of a triangulation gives a colex on any manifold in Sec.~\ref{sec:barycentric_subdivision}. Finally, we introduce the hyperbolic manifolds proposed in Ref.~\cite{Guth:2014zxf} in Sec.~\ref{sec:hyperbolic-manifold}, from which hyperbolic toric codes with constant rate and polynomial distance are obtained.

\subsection{Colexes}
\label{sec:colexes}

To relate color codes to toric codes, we introduce the notion of a \emph{colex}, following Ref.~\cite{Kubica:2019qji}. We note, however, that the colex structure is not required merely to ensure the commutativity of the color-code stabilizers; for this purpose, it is sufficient to consider a \(D\)-dimensional simplicial complex whose vertices admit a proper \((D+1)\)-coloring~\cite{Delfosse:2013smk, Kubica:2014jue}. 

A \(D\)-colex is defined recursively in terms of an auxiliary notion of a \emph{colorable ball}, which we introduce first. As a base case, a colorable \(0\)-ball is defined to be a single vertex. For \(D\ge 1\), a colorable \(D\)-ball \(\mathcal B_{v,Y}\) is then defined as the set of \(D\)-simplices spanned by a vertex \(v\) and all the \((D-1)\)-simplices of some \((D-1)\)-colex \(Y\) homeomorphic to a \((D-1)\)-sphere,
\begin{align}
    \mathcal B_{v,Y} \coloneqq  \{v \star \tau \mid \tau \in \Delta_{D-1}(Y)\}.
\end{align}
The colorable \(D\)-ball \(\mathcal B_{v,Y}\) is \((D+1)\)-colorable, with \(v\) assigned a color distinct from the \(D\) colors used for the vertices of \(Y\). Since \(Y\) is homeomorphic to \(S^{D-1}\), the \(D\)-simplices \(\{v\star\tau\}_{\tau\in\Delta_{D-1}(Y)}\) glue consistently along their shared faces, so that \(\mathcal{B}_{v,Y}\) is itself a simplicial complex, homeomorphic to a \(D\)-ball with boundary \(\partial\mathcal{B}_{v,Y} \cong Y\).

A \(0\)-colex is defined to be a collection of colorable \(0\)-balls. A \(D\)-colex \(X\) is a union of disjoint colorable \(D\)-balls \(\{\mathcal B_{v,Y}\}\),
\begin{align}
    X \coloneqq  \bigsqcup_{v, Y} \mathcal B_{v,Y}.
\end{align}
glued together along their \((D-1)\)-dimensional boundaries \(\partial \mathcal B_{v,Y}\), such that every \((D-1)\)-simplex in \(\partial\mathcal B_{v,Y}\) belongs to the boundary of exactly two distinct colorable \(D\)-balls, and such that the vertices of \(X\) are \((D+1)\)-colorable.

Since a \(D\)-colex is \((D+1)\)-colorable, and any two vertices sharing an edge are assigned different colors, the \(i+1\) vertices of an \(i\)-simplex \(\sigma\) carry \(i+1\) distinct colors. We write the set of these colors as
\begin{align}
    \mathrm{color}(\sigma) \coloneqq  \{\mathrm{color}(v)\mid v\in \Delta_0(\sigma)\}.
\end{align}

\subsection{Barycentric subdivision}
\label{sec:barycentric_subdivision}

We now show that a colex can be obtained explicitly from any triangulation of a closed manifold via \emph{barycentric subdivision}. This construction allows us to define color codes on general closed manifolds using colexes.

For an \(i\)-simplex \(\sigma\) with vertex set \(\{v_0,\ldots,v_i\}\), we write \(b(\sigma) \coloneqq  \frac{1}{i+1}\sum_{j=0}^i v_j\) for its \emph{barycenter}. We define the barycentric subdivision \(\mathrm{sd}(\sigma)\) of an \(i\)-simplex \(\sigma\) recursively on \(i\). As a base case, \(\mathrm{sd}(\Delta^0) \coloneqq  \Delta^0\). For \(i\ge 1\), we first define \(\mathrm{sd}(\partial\sigma)\) by applying this construction to each \((i-1)\)-simplex of \(\partial\sigma\) and gluing along shared faces, and then set
\begin{align}
    \mathrm{sd}(\sigma) \coloneqq  b(\sigma) \star \mathrm{sd}(\partial\sigma).
\end{align}
For a \(D\)-dimensional simplicial complex \(X\), we define \(\mathrm{sd}(X) \coloneqq  \bigcup_{\delta\in\Delta_D(X)} \mathrm{sd}(\delta)\), again glued along shared faces.

The recursive structure of barycentric subdivision matches the recursive construction of a colex via colorable balls: for each \(D\)-simplex \(\delta\in\Delta_D(X)\), the formula \(\mathrm{sd}(\delta) = b(\delta)\star\mathrm{sd}(\partial\delta)\) is precisely the colorable ball \(\mathcal{B}_{b(\delta),\,\mathrm{sd}(\partial\delta)}\), with \(\mathrm{sd}(\partial\delta)\) itself a \((D-1)\)-colex by the same construction applied one dimension lower. This yields the following.

\begin{proposition}[Barycentric subdivision yields a colex]
\label{prop:barycentric_subdivision_derive_colex}
Let \(X\) be a triangulation of a closed \(D\)-manifold. Then \(\mathrm{sd}(X)\) is a \(D\)-colex, with \(\mathcal{B}_{b(\delta),\,\mathrm{sd}(\partial\delta)}\) for \(\delta\in\Delta_D(X)\) as its colorable \(D\)-balls, and with the coloring of \(\mathrm{sd}(X)\) given by assigning to each vertex \(b(\sigma)\) (for \(\sigma\in\Delta_i(X)\)) the color \(i\).
\end{proposition}

\begin{proof}
We proceed by induction on \(D\).

\textit{Base case (\(D=0\)).} A closed \(0\)-manifold is a finite set of points, and \(\mathrm{sd}(X) = X\), a disjoint union of colorable \(0\)-balls, i.e., a \(0\)-colex by definition. The coloring assigning color \(0\) to every vertex is trivially proper, as a \(0\)-colex has no edges.

\textit{Inductive step.} Let \(D\ge 1\) and assume the proposition holds for all dimensions less than \(D\). Let \(X\) be a triangulation of a closed \(D\)-manifold. For each \(\delta\in\Delta_D(X)\), the boundary \(\partial\delta\) is itself a triangulation of the closed \((D-1)\)-manifold \(S^{D-1}\). By the inductive hypothesis applied to \(\partial\delta\), \(\mathrm{sd}(\partial\delta)\) is a \((D-1)\)-colex homeomorphic to \(S^{D-1}\), with colorable \((D-1)\)-balls \(\mathcal{B}_{b(\mu),\,\mathrm{sd}(\partial\mu)}\) for \(\mu\in\Delta_{D-1}(\delta)\), and coloring \(b(\tau)\mapsto\dim\tau\) for \(\tau\) a proper face of \(\delta\). Hence \(\mathrm{sd}(\delta) = \mathcal{B}_{b(\delta),\,\mathrm{sd}(\partial\delta)}\) is a well-defined colorable \(D\)-ball, with boundary \(\partial(\mathrm{sd}(\delta)) = \mathrm{sd}(\partial\delta)\).

By definition, \(\mathrm{sd}(X) = \bigcup_{\delta\in\Delta_D(X)} \mathrm{sd}(\delta)\), glued along shared faces. Since \(X\) triangulates a closed \(D\)-manifold, every \((D-1)\)-simplex \(\mu\in\Delta_{D-1}(X)\) is a face of exactly two \(D\)-simplices \(\delta_1,\delta_2\in\Delta_D(X)\), and \(\mathrm{sd}(\mu)\) depends only on \(\mu\) itself, giving rise to a piece common to \(\partial(\mathrm{sd}(\delta_1))\) and \(\partial(\mathrm{sd}(\delta_2))\). Hence every \((D-1)\)-simplex of \(\mathrm{sd}(X)\) belongs to the boundary of exactly two distinct colorable \(D\)-balls, as required.

Finally, coloring the vertices of \(\mathrm{sd}(X)\) by \(b(\sigma)\mapsto\dim\sigma\) for \(\sigma\in\Delta_i(X)\), \(0\le i\le D\), gives a \((D+1)\)-coloring, which is proper since any edge of \(\mathrm{sd}(X)\) connects the barycenters of two comparable faces of \(X\), which necessarily have different dimensions. This completes the induction and proves the proposition.
\end{proof}

\subsection{Hyperbolic manifolds}
\label{sec:hyperbolic-manifold}

 We construct an explicit family of closed hyperbolic manifolds, following the construction of Ref.~\cite{Guth:2014zxf}. These manifolds will serve as the underlying cell complexes on which we define color codes and toric codes in the following sections.

The \emph{\(D\)-dimensional hyperbolic space} is defined as
\begin{align}
    \mathbb{H}^{D} \coloneqq  \left\{(x_i)^D_{i=0} \in \mathbb{R}^{D+1} \ \middle|\ -x_0^2 + \sum^D_{i=1} x_i^2 = -1,~ x_0 >0 \right\}.
\end{align}
Let \(\Gamma\) be a discrete group of isometries acting freely on \(\mathbb{H}^D\), and define an equivalence relation \(\sim\) on \(\mathbb{H}^D\) by \(x\sim y\) if and only if \(y=\gamma\cdot x\) for some \(\gamma\in\Gamma\). A hyperbolic manifold can then be obtained as the quotient of \(\mathbb{H}^D\) by \(\Gamma\),
\begin{align}
    \Gamma\backslash \mathbb{H}^D \coloneqq  \mathbb{H}^D / {\sim}.
\end{align}
Here, \(\backslash\) denotes the quotient by the left action of \(\Gamma\), rather than set subtraction.
Such a quotient is not closed for a generic choice of \(\Gamma\). We first construct a cocompact arithmetic lattice, and then pass to sufficiently deep torsion-free congruence subgroups.

Let \(f\colon \mathbb{R}^{D+1} \to \mathbb{R}\) be the quadratic form \(-\sqrt{2}x_0^2 + \sum^D_{i=1} x_i^2\). We define \(SO_f\) as the group of \((D+1)\times(D+1)\) matrices with entries in \(\mathbb{R}\) that preserve \(f\) and have determinant \(1\),
\begin{align}
    SO_f \coloneqq  &\Big\{ A \in \mathbb{R}^{(D+1)\times (D+1)} 
    \ \Big| \notag \\
    & f(Ax) = f(x), \ \forall x\in\mathbb{R}^{D+1},\ \det A = 1 \Big\}.
\end{align}
Let \(SO_f(\mathbb{Z}[\sqrt{2}]) \subset SO_f\) be the subgroup of matrices with entries in \(\mathbb{Z}[\sqrt{2}]\), where \(\mathbb{Z}[\sqrt 2] = \{a+b\sqrt 2 \mid a,b \in \mathbb{Z} \}\). The group \(SO_f(\mathbb{Z}[\sqrt 2])\) is a discrete subgroup of \(SO_f\), which provides a closed hyperbolic manifold; see Ref.~\cite[12.8 Example 5]{ratcliffeFoundationsHyperbolicManifolds2019} for the fact that the resulting quotient is compact.

To construct a family of closed manifolds, we define a sequence of subgroups of \(SO_f(\mathbb{Z}[\sqrt 2])\). We let \(\Gamma_1 \coloneqq  SO_f(\mathbb{Z}[\sqrt{2}])\). Since \(f\) has the same signature \((D,1)\) as the form defining \(\mathbb{H}^D\), it is linearly congruent to that form, and we fix an identification of \(SO_f\) with \(SO(D,1)\) accordingly; in particular, \(\Gamma_1\) and its subgroups act on \(\mathbb{H}^D\) by isometries. For \(N\ge 1\), define
\begin{align}
\label{eq:principal_congruence_subgroups}
    &\Gamma_N \coloneqq  \notag \\
    &\left\{A+B\sqrt{2} \in \Gamma_1 \ \middle|\ A\equiv I \!\!\!\pmod N,\ B\equiv 0\!\!\! \pmod N\right\},
\end{align}
where \(A\) and \(B\) are matrices with integer entries. The subgroups \(\Gamma_N \subset \Gamma_1\) are called \emph{principal congruence subgroups} of \(\Gamma_1\). In order for the quotient \(\Gamma_N\backslash\mathbb{H}^D\) to be a manifold rather than an orbifold, \(\Gamma_N\) must act on \(\mathbb{H}^D\) freely, i.e., without fixed points. This is indeed the case for sufficiently large \(N\): the quotients \(M_N\coloneqq \Gamma_N\backslash\mathbb{H}^{D}\) are closed hyperbolic manifolds for all sufficiently large \(N\)~\cite{Guth:2014zxf}. We fix one such sufficiently large \(N_0\) once and for all, and let \(M_0 \coloneqq  \Gamma_{N_0}\backslash\mathbb{H}^D\) denote the resulting closed hyperbolic manifold, which will serve as the base of the tower of finite covers \(\{M_N\}_{N_0\mid N}\) considered below; since \(\Gamma_N\subset\Gamma_{N_0}\) whenever \(N_0\mid N\), each such \(M_N\) is indeed a finite-sheeted cover of \(M_0\).

Combining the constructions introduced in this section, we can now describe explicitly how to construct color codes on the family of closed hyperbolic manifolds \(M_N\). We first fix a triangulation \(Y_0\) of the base manifold \(M_0\), and let
\begin{align}
X_0\coloneqq \mathrm{sd}(Y_0).
\end{align}
By Proposition~\ref{prop:barycentric_subdivision_derive_colex}, \(X_0\) is a \(D\)-colex. For each \(N\) divisible by \(N_0\), let \(p_N\colon M_N\to M_0\) denote the corresponding finite-sheeted covering map, and define \(X_N\coloneqq p_N^{-1}(X_0)\) to be the pullback triangulation of \(M_N\). Since a covering map is a local homeomorphism, the simplicial structure and vertex coloring of \(X_0\) lift to \(X_N\); hence \(X_N\) is again a \(D\)-colex. Finally, for a type index \(1\le i\le D-1\), we define the color-code chain complex as in Eq.~\eqref{eq:color-code-chain-complex},
\begin{align}
   & CC_i(X_N)= \notag\\
   & \left(
   \begin{CD}
       C_{D-i-1}(X_N) @>{\partial_{D-i-1, D}}>> C_D(X_N) @>{\partial_{D,i-1}}>> C_{i-1}(X_N)
   \end{CD}
   \right),
\end{align}
thereby defining a color code on \(M_N\). In the following sections, we study how the parameters of this family of color codes scale with \(N\) and with the dimension \(D\).

\section{Code Distances of Color and Toric Codes}
\label{sec:code-distances-of-color-and-toric-codes}

In this section, we derive a lower bound on the code distance of the hyperbolic color code by relating it to the code distance of an associated toric code. First, we introduce the morphism between color and toric codes, following Ref.~\cite{Kubica:2019qji}, in Sec.~\ref{sec:color-to-toric-morphism}. This morphism enables us to decompose a color code into multiple toric codes. We then show that the code distance of the toric codes obtained by the morphism can be bounded below by a polynomial in \(n\) in Sec.~\ref{sec:distance_lower_bound_of_hyperbolic_toric_code}. Finally, we show that the code distance of the color code can be bounded below by that of the toric codes obtained by the morphism in Sec.~\ref{sec:color-code-lower-bound}.

\subsection{Morphism between color and toric codes}
\label{sec:color-to-toric-morphism}

Following Ref.~\cite{Kubica:2019qji}, we introduce a morphism relating color codes and toric codes. This morphism is defined in terms of the \emph{restricted lattice}~\cite{Bombin:2006rk, Kubica:2019qji}, obtained from a colex by removing its vertices of a chosen color.

To define the restricted lattice, we first introduce the notion of a link. For any simplex \(\sigma\) in a \(D\)-colex \(X\) and any \(i\ge0\), the \emph{\(i\)-link} of \(\sigma\) is
\begin{align}
    &\mathrm{Lk}_i(\sigma) \coloneqq  \notag \\ 
    &\{\tau\in\Delta_i(X) \mid \tau\cap\sigma = \emptyset \text{ and } \exists\delta\in\mathrm{St}_D(\sigma) \text{ s.t. } \tau\in\Delta_i(\delta)\}.
\end{align}

We are now ready to define the restricted lattice. Let \(X\) be a \(D\)-colex, and let \(C\subset \mathbb{Z}_{D+1}\) be a subset of \(i+1\) colors, where \(1\le i < D\). The \emph{restricted lattice} \(X_C\) is an \((i+1)\)-dimensional cell complex constructed from \(X\) as follows.
\begin{itemize}
    \item For \(0\le j \le i\), the \(j\)-cells of \(X_C\) coincide with the \(j\)-simplices of \(X\) whose color set is contained in \(C\),
    \[
    \Lambda_j(X_C) = \{\sigma\in \Delta_j(X) \mid \mathrm{color}(\sigma)\subset C\}.
    \]
    \item Each \((i+1)\)-cell \(\Xi(\sigma)\) of \(X_C\) is obtained by attaching an \((i+1)\)-disk, along its boundary, to the \(i\)-link \(\mathrm{Lk}_i(\sigma)\) of a \((D-i-1)\)-simplex \(\sigma\) of \(X\) with color \(\mathbb{Z}_{D+1}\setminus C\); the resulting cell \(\Xi(\sigma)\) is thus uniquely identified with \(\sigma\).
\end{itemize}
The restricted lattice \(X_C\) is a regular cell complex: removing from \(X\) the vertices of color \(\mathbb{Z}_{D+1}\setminus C\) leaves a cavity around each such vertex, bounded by its link, and \(X_C\) is obtained precisely by filling in each such cavity with a single \((i+1)\)-cell.

Having constructed the restricted lattice \(X_C\), we now define a linear map relating the chain spaces of \(X\) and \(X_C\); this map will allow us to relate the color code on \(X\) to a toric code on \(X_C\). The \emph{restriction} \(\pi_C\) is a triple of linear maps \(\pi^{(0)}_C\colon C_{i-1}(X) \to C_{i-1}(X_C)\), \(\pi^{(1)}_C\colon C_D(X) \to C_i(X_C)\), and \(\pi^{(2)}_C\colon C_{D-i-1}(X) \to C_{i+1}(X_C)\), defined as
\begin{align}
    \pi^{(0)}_C(\sigma) &\coloneqq  \begin{cases}
        \sigma & \text{if }\mathrm{color}(\sigma)\subset C
        , \\
        0 & \text{otherwise},
    \end{cases} \\
    \pi^{(1)}_C(\tau) &\coloneqq  \tau_C, \\
    \pi^{(2)}_C(\upsilon) &\coloneqq  \begin{cases}
        \Xi(\upsilon) & \text{if }\mathrm{color}(\upsilon) = \mathbb{Z}_{D+1}\setminus C, \\
        0 & \text{otherwise},
    \end{cases}
\end{align}
where \(\tau_C\) denotes the \(i\)-simplex of color \(C\) belonging to the \(D\)-simplex \(\tau\).

Let \(CC_i(X)\) and \(TC_i(X_C)\) denote the chain complexes of the color code on \(X\) and the toric code on the restricted lattice \(X_C\), respectively. The restriction map \(\pi_C\) defines a morphism between these chain complexes, as expressed by the following commutative diagram.
\begin{align}
\label{eq:commutativity_of_morphisms}
    \begin{CD}
        C_{D-i-1}(X) @>{\partial_{D-i-1, D}}>> C_D(X) @>{\partial_{D,i-1}}>> C_{i-1}(X) \\
        @VV{\pi^{(2)}_C}V @VV{\pi^{(1)}_C}V @VV{\pi^{(0)}_C}V \\
        C_{i+1}(X_C) @>{\partial_{i+1,i}^C}>>C_i(X_C) @>{\partial_{i,i-1}^C}>> C_{i-1}(X_C)
    \end{CD}
\end{align}

Fix a color \(c^*\in\mathbb{Z}_{D+1}\), and let \(\mathcal C_{i,c^*}\) denote the collection of \(i\)-element color sets not containing \(c^*\),
\begin{align}
\label{eq:color_set}
    \mathcal C_{i,c^*} \coloneqq  \{C\subset \mathbb{Z}_{D+1} \mid c^* \notin C \land |C|=i\}.
\end{align}
By Ref.~\cite[Section 4.2]{Kubica:2019qji}, the first homology group of the color code chain complex \(CC_i(X)\) is isomorphic to the product of the first homology groups of the toric code chain complexes \(TC_i(X_C)\) on the restricted lattices \(X_C\) for \(C\in\mathcal C_{i,c^*}\),
\begin{align}
    H_1(CC_i(X)) \cong \prod_{C\in \mathcal C_{i,c^*}} H_1(TC_i(X_{C^*})).
\end{align}
where \(C^*=C \sqcup \{c^*\}\).
In other words, the color code on \(X\) is homologically equivalent to \(|\mathcal{C}_{i,c^*}|=\binom{D}{i}\) copies of the toric code on the restricted lattices \(X_{C^*}\), one for each \(C\in\mathcal{C}_{i,c^*}\). This decomposition is consistent with the number of logical qubits of the color code given in Eq.~\eqref{eq:num_of_logical_qubits_color_code}. By Eq.~\eqref{eq:num_of_logical_qubits_toric_code}, each toric code \(TC_i(X_{C^*})\) encodes \(b_i(X_{C^*})\) logical qubits, and \(b_i(X_{C^*})=b_i(X)\) for every \(C\in\mathcal{C}_{i,c^*}\). Since there are \(\binom{D}{i}\) such copies, the total number of logical qubits is \(k_{CC_i(X)}=\binom{D}{i}b_i(X).\)

\subsection{A lower bound on the distance of hyperbolic toric codes}
\label{sec:distance_lower_bound_of_hyperbolic_toric_code}

We now turn to establishing a lower bound on the code distance of the toric code defined on a hyperbolic manifold, in terms of the systole of the underlying manifold. To define the systole, we briefly recall the notion of the volume of a homology class; see, e.g., Ref.~\cite{Guth:2014zxf} for details. Let \(M\) be a smooth manifold. A \emph{Lipschitz \(i\)-chain} with coefficients in \(\mathbb{Z}_2\) is a finite sum \(\sum_j a_j f_j\), where \(a_j\in\mathbb{Z}_2\) and each \(f_j\) is a Lipschitz map from the standard \(i\)-simplex \(\Delta^i\) to \(M\); we write \(C_{i,\mathrm{Lip}}(M,\mathbb{Z}_2)\) for the resulting chain group, with boundary map \(\partial_{i,\mathrm{Lip}}\) defined by restricting each \(f_j\) to the \((i-1)\)-simplices in \(\partial\Delta^i\), as usual. By Rademacher's theorem~\cite[Theorem 3.1.6]{federerGeometricMeasureTheory1996}, a Lipschitz map \(f_j\colon \Delta^i\to (M,g)\) is differentiable almost everywhere, so that the pullback \(f_j^*g\) is defined almost everywhere on \(\Delta^i\); this allows us to assign to \(f_j\) a volume \(\vol_i(f_j) \coloneqq  \int_{\Delta^i}\sqrt{\det(f_j^*g)}\), and to a chain \(\sum_j a_j f_j\) the volume \(\vol_i\big(\sum_j a_j f_j\big) \coloneqq  \sum_j |a_j|\,\vol_i(f_j)\).

Given a Riemannian manifold \((M,g)\), the \(i\)-dimensional \emph{systole} of \((M,g)\) with coefficients in \(\mathbb{Z}_2\) is the infimal volume of a homologically nontrivial Lipschitz \(i\)-cycle,
\begin{align}
    \sys_i(M,g) \coloneqq  \inf_{\alpha\in \ker \partial_{i,\mathrm{Lip}}\setminus \im \partial_{i+1,\mathrm{Lip}}} \vol_i(\alpha).
\end{align}

The systole gives a lower bound on the code distance of a toric code defined on a closed manifold. While the bound established in Ref.~\cite{Guth:2014zxf} applies to a toric code defined on a triangulation of the manifold, we will need the corresponding bound for a toric code defined on a restricted lattice.

\begin{proposition}[Toric-code distance on restricted lattices bounded below by the systole]
\label{prop:restricted_toric_systole_bound}
Let \(Y_0\) be a finite smooth triangulation of the closed Riemannian manifold \((M_0, g_0)\), and let \(X_0 = \mathrm{sd}(Y_0)\) be its barycentric subdivision. Let \(M\to M_0\) be a finite-sheeted cover, and let \(g\) and \(X\) be the pullbacks of \(g_0\) and \(X_0\), respectively. Then, for every \(1\le j < D\) and every set \(C\subset \mathbb{Z}_{D+1}\) of \(j+1\) colors,
\begin{align}
    d^Z_{TC_j(X_C)} \ \ge\ c_1(M_0,g_0,X_0, j)\,\sys_j(M,g),
\end{align}
where \(c_1(M_0,g_0,X_0, j)>0\) is independent of the cover.
\end{proposition}

Since our goal is a lower bound on the distance of the color code, and since Theorem~\ref{th:code_distance_lower_bound_XZ} bounds both the \(Z\)-type and \(X\)-type contributions to that distance solely in terms of the minimal weight of nontrivial \emph{logical \(Z\)} operators of toric codes on restricted lattices, it suffices to bound \(d^Z_{TC_i(X_C)}\); no bound on the full toric-code distance \(d_{TC_i(X_C)} = \min(d^Z_{TC_i(X_C)}, d^X_{TC_i(X_C)})\) is needed for this purpose.

Our strategy for proving Proposition~\ref{prop:restricted_toric_systole_bound} is to relate nontrivial logical \(Z\) operators of the toric code \(TC_i(X_C)\) on the restricted lattice to the underlying colex \(X\) itself: we show that every nontrivial \(i\)-cycle of \(TC_i(X_C)\) is, when viewed as an \(i\)-chain for \(X\), again a nontrivial \(i\)-cycle for \(X\). The following lemma establishes this.

\begin{lemma}[Nontrivial cycles of a restricted lattice remain nontrivial in the colex]
\label{lem:restricted_cycles}
Let \(X\) be a \(D\)-colex, let \(1\le i<D\), and let \(C\subset\mathbb Z_{D+1}\) be a set of \(i+1\) colors. Then
\begin{align}
    \ker\partial^C_{i,i-1}\setminus\im\partial^C_{i+1,i} \ \subseteq\ \ker\partial_{i,i-1}\setminus\im\partial_{i+1,i},
\end{align}
where the operators on the left are those of \(X_C\) and those on the right are those of \(X\).
\end{lemma}
\begin{proof}
By construction of the restricted lattice, \(C_i(X_C)\) is the subspace of \(C_i(X)\) spanned by the \(i\)-simplices of color \(C\), and every face of such a simplex again has color contained in \(C\); hence \(\partial^C_{i,i-1}\) is the restriction of \(\partial_{i,i-1}\) to \(C_i(X_C)\), and \(\ker\partial^C_{i,i-1}\subseteq\ker\partial_{i,i-1}\). Moreover, part (ii) of the proof of Theorem 1 of Ref.~\cite{Kubica:2019qji} establishes
\begin{align}
    \ker\partial^C_{i,i-1}\cap\im\partial_{i+1,i} = \im\partial^C_{i+1,i},
\end{align}
so an element of \(\ker\partial^C_{i,i-1}\) lying in \(\im\partial_{i+1,i}\) would already lie in \(\im\partial^C_{i+1,i}\).
\end{proof}

\begin{proof}[Proof of Proposition~\ref{prop:restricted_toric_systole_bound}]
Let \(\alpha\in\ker\partial^C_{j,j-1}\setminus\im\partial^C_{j+1,j}\). By Lemma~\ref{lem:restricted_cycles}, \(\alpha\) is a \(j\)-cycle in \(X\) that is not a \(j\)-boundary in \(X\). Since \(X\) is a smooth triangulation of \(M\), the characteristic map of each simplex is Lipschitz, so \(\alpha\in C_{j,\mathrm{Lip}}(M,\mathbb Z_2)\) with \(\partial_{j,\mathrm{Lip}}\alpha=\partial_{j,j-1}\alpha=0\); since simplicial and Lipschitz homology of \(M\) with \(\mathbb Z_2\) coefficients agree, \(\alpha\in\ker\partial_{j,\mathrm{Lip}}\setminus\im\partial_{j+1,\mathrm{Lip}}\). Hence
\begin{align}
    \vol_j(\alpha)\ \ge\ \sys_j(M,g).
\end{align}
On the other hand, \(\alpha\) consists of \(\wt(\alpha)\) simplices, each of volume at most \(\max_{\tau\in\Delta_j(X)}\vol_j(\tau)\), which equals \(\max_{\tau\in\Delta_j(X_0)}\vol_j(\tau)\) because the covering map is a local isometry and \(X\) is the pullback of \(X_0\). Therefore \(\vol_j(\alpha)\le (\max_{\tau\in\Delta_j(X_0)}\vol_j(\tau))\wt(\alpha)\), and combining the two displays gives \(\wt(\alpha)\ge c_1(M_0,g_0,X_0,j)\sys_j(M,g)\). Taking the minimum over \(\alpha\) yields the claim.
\end{proof}

This proposition reduces the problem of lower-bounding the code distance to that of lower-bounding the systole of the underlying manifold. In Sec.~\ref{sec:explicit-lower-bounds}, we apply this result to the specific family of hyperbolic manifolds \(M_N\) constructed in Sec.~\ref{sec:hyperbolic-manifold}, and derive an explicit polynomial lower bound on the code distance in terms of the number of physical qubits.

\subsection{Bounding the color-code distance with the toric code distance}
\label{sec:color-code-lower-bound}

Having established a lower bound on the toric code distance in terms of the systole, we now turn to our main result: a lower bound on the color-code distance in terms of the toric code distance.

\begin{theorem}[Color-code distance from toric-code distance]
\label{th:code_distance_lower_bound_XZ}
Let \(X\) be a \(D\)-colex. Then the \(Z\)-type and \(X\)-type contributions to the code distance of \(CC_i(X)\) satisfy
\begin{align}
    d^{Z}_{CC_i(X)} &\ge \max_{c^*\in\mathbb{Z}_{D+1}} \ \min_{C\in\mathcal C_{i,c^*}} \ d^Z_{TC_i(X_{C\sqcup\{c^*\}})}, \\
    d^{X}_{CC_i(X)} &\ge \max_{c^*\in\mathbb{Z}_{D+1}} \ \min_{C\in\mathcal C_{D-i,c^*}} \ d^Z_{TC_{D-i}(X_{C\sqcup\{c^*\}})},
\end{align}
where \(d^Z_{TC_j(\mathcal L)} \coloneqq  \min_{\alpha\in\ker\partial^{\mathcal L}_{j,j-1}\setminus\im\partial^{\mathcal L}_{j+1,j}}\wt(\alpha)\) denotes the minimal weight of a nontrivial logical \(Z\) operator of the type-\(j\) toric code on \(\mathcal L\).
\end{theorem}

We prove this theorem using two lemmas. The first shows that any nontrivial logical \(Z\) operator of the color code restricts, under the morphism \(\pi_C\) for some choice of colors \(C\), to a nontrivial logical \(Z\) operator of the corresponding toric code. The second shows that this restriction does not increase the weight of the operator.

\begin{lemma}[A nontrivial color-code class can be restricted to a nontrivial toric-code class]
\label{lem:nontrivial_cycle_under_morphism}
    Let \(X\) be a \(D\)-colex, and fix a color \(c^*\in\mathbb{Z}_{D+1}\). Let \(\mathcal C_{i,c^*}\) be the collection of subsets of \(\mathbb{Z}_{D+1}\setminus\{c^*\}\) of size \(i\). Then, for every nontrivial class \([\lambda]\in \ker\partial_{D,i-1}/\im\partial_{D-i-1,D}\) on \(X\), there exists a color set \(C\in\mathcal C_{i,c^*}\) such that
    \begin{align}
        [\pi^{(1)}_{C\sqcup \{c^*\}} (\lambda)] \neq 0
    \end{align}
    in \(H_1(TC_i(X_{C\sqcup\{c^*\}}))\).
\end{lemma}

\begin{proof}
    By Ref.~\cite[Section 4.2]{Kubica:2019qji}, the map
    \begin{align}
        \pi \coloneqq  \big(\pi^{(1)}_{C\sqcup\{c^*\}}\big)_{C\in\mathcal C_{i,c^*}}
    \end{align}
    is an isomorphism from \(H_1(CC_i(X))\) to \(\prod_{C\in\mathcal C_{i,c^*}} H_1(TC_i(X_{C\sqcup\{c^*\}}))\). In particular, \(\pi\) is injective, so \([\lambda]\neq 0\) implies \(\pi([\lambda])\neq 0\) in the product. Since an element of a product of vector spaces is nonzero if and only if at least one of its components is nonzero, there exists \(C\in\mathcal C_{i,c^*}\) such that the \(C\)-component of \(\pi([\lambda])\), namely \([\pi^{(1)}_{C\sqcup\{c^*\}}(\lambda)]\), is nonzero.
\end{proof}

\begin{lemma}[The restriction map does not increase weight]
\label{lem:weight_bound}
Let \(X\) be a \(D\)-colex, and let \(C\subset\mathbb Z_{D+1}\) be a set of \(i+1\) colors. Then, for every \(D\)-chain \(\lambda\in C_D(X)\),
\begin{align}
    \mathrm{wt}\bigl(\pi_C^{(1)}(\lambda)\bigr)\le\mathrm{wt}(\lambda).
\end{align}
\end{lemma}
\begin{proof}
    Write \(\lambda = \sum_{\delta\in\Delta_D(X)} a_\delta\,\delta\) with \(a_\delta\in\mathbb F_2\), and let \(\operatorname{supp}(\lambda) \coloneqq  \{\delta\in\Delta_D(X)\mid a_\delta=1\}\), so that \(\mathrm{wt}(\lambda) = |\operatorname{supp}(\lambda)|\). Recall that \(\pi_C^{(1)}\) sends each \(D\)-simplex \(\delta\) to the unique \(i\)-simplex \(\delta_C\) of color \(C\) belonging to \(\delta\), so that
    \begin{align}
        \pi_C^{(1)}(\lambda) = \sum_{\delta\in\operatorname{supp}(\lambda)} \delta_C.
    \end{align}
    Every term appearing in this sum lies in the image of \(\operatorname{supp}(\lambda)\) under \(\delta\mapsto\delta_C\), and hence
    \begin{align}
        \operatorname{supp}\bigl(\pi_C^{(1)}(\lambda)\bigr) \subseteq \{\delta_C\mid \delta\in\operatorname{supp}(\lambda)\}.
    \end{align}
    Since the right-hand side has at most \(|\operatorname{supp}(\lambda)|\) elements, we conclude
    \begin{align}
        \mathrm{wt}\bigl(\pi_C^{(1)}(\lambda)\bigr) &= \bigl|\operatorname{supp}\bigl(\pi_C^{(1)}(\lambda)\bigr)\bigr| \\
        &\le\bigl|\{\delta_C\mid \delta\in\operatorname{supp}(\lambda)\}\bigr| \\
        &\le|\operatorname{supp}(\lambda)|\\
        &=\mathrm{wt}(\lambda),
    \end{align}
    as claimed.
\end{proof}

\begin{proof}[Proof of Theorem~\ref{th:code_distance_lower_bound_XZ}]
    If \(H_1(CC_i(X))=0\), then there is no nontrivial logical \(Z\) operator, \(d^Z_{CC_i(X)}\) is undefined (or conventionally infinite), and the bound holds trivially. Assume henceforth \(H_1(CC_i(X))\neq 0\). Since \(X\) is a finite colex, the set \(\ker\partial_{D,i-1}\setminus\im\partial_{D-i-1,D}\) is finite and nonempty, so the minimum defining \(d^Z_{CC_i(X)}\) is attained; choose \(\lambda_0\) realizing it, so that \(d^Z_{CC_i(X)} = \wt(\lambda_0)\) and \([\lambda_0]\neq 0\) in \(H_1(CC_i(X))\).

    Fix a color \(c^*\in\mathbb{Z}_{D+1}\). By Lemma~\ref{lem:nontrivial_cycle_under_morphism}, there exists \(C_0\in\mathcal C_{i,c^*}\) such that
    \begin{align}
        [\pi^{(1)}_{C_0\sqcup\{c^*\}}(\lambda_0)] \neq 0 \quad \text{in } H_1(TC_i(X_{C_0\sqcup\{c^*\}})),
    \end{align}
    i.e., \(\pi^{(1)}_{C_0\sqcup\{c^*\}}(\lambda_0)\) is a nontrivial logical \(Z\) operator of the type-\(i\) toric code on the restricted lattice \(X_{C_0\sqcup\{c^*\}}\). By definition of the toric-code distance,
    \begin{align}
        \wt\bigl(\pi^{(1)}_{C_0\sqcup\{c^*\}}(\lambda_0)\bigr) \ge d_{TC_i(X_{C_0\sqcup\{c^*\}})}^Z.
    \end{align}
    On the other hand, Lemma~\ref{lem:weight_bound} gives
    \begin{align}
        \wt(\lambda_0) \ge \wt\bigl(\pi^{(1)}_{C_0\sqcup\{c^*\}}(\lambda_0)\bigr).
    \end{align}
    Combining the two inequalities, and using that \(C_0\) is one particular element of \(\mathcal C_{i,c^*}\),
    \begin{align}
        d^Z = \wt(\lambda_0) \ge d^Z_{TC_i(X_{C_0\sqcup\{c^*\}})} \ge \min_{C\in\mathcal C_{i,c^*}} d^Z_{TC_i(X_{C\sqcup\{c^*\}})}.
    \end{align}
    Since \(c^*\in\mathbb{Z}_{D+1}\) was an arbitrary fixed color, this bound holds for every choice of \(c^*\), and hence
    \begin{align}
        d^Z \ge \max_{c^*\in\mathbb{Z}_{D+1}} \min_{C\in\mathcal C_{i,c^*}} d^Z_{TC_i(X_{C\sqcup\{c^*\}})}.
    \end{align}

    By duality, the same argument applied to \(H_1(CC_{D-i}(X))\) yields
    \begin{align}
        d^X \ge \max_{c^*\in\mathbb{Z}_{D+1}} \min_{C\in\mathcal C_{D-i,c^*}} d^Z_{TC_{D-i}(X_{C\sqcup\{c^*\}})}
    \end{align}
    as well.
\end{proof}

In the corollary below and for the remainder of the paper, we abbreviate \(c_1(M_0,g_0,X_0,j)\) from Proposition~\ref{prop:restricted_toric_systole_bound} as \(c_1(j)\), suppressing the dependence on \(M_0,g_0,X_0\) where it is clear from the context.

\begin{corollary}[Color-code distance bounded below by the systole]
\label{cor:color-code-distance-bounded-by-the-systole}
With the notation of Proposition~\ref{prop:restricted_toric_systole_bound}, the color code on \(X\) satisfies
\begin{align}
    d_{CC_i(X)} \ge\min\bigl(c_1(i)\sys_i(M,g),c_1(D-i) \sys_{D-i}(M,g)\bigr).
\end{align}
\end{corollary}
\begin{proof}
Apply Theorem~\ref{th:code_distance_lower_bound_XZ} and then Proposition~\ref{prop:restricted_toric_systole_bound} with \(j=i\) to bound \(d^Z\), and with \(j=D-i\) to bound \(d^X\); note \(1\le D-i<D\). The results follows from \(d_{CC_i(X)}=\min(d^Z,d^X)\).
\end{proof}

The corollary above reduces the problem of lower-bounding the color-code distance to the geometry of the underlying manifold. For any sequence of finite covers \(M\to M_0\), lower bounds on \(\sys_i(M,g)\) and \(\sys_{D-i}(M,g)\) yield a lower bound on \(d_{CC_i(X)}\), with a constant depending only on the fixed \((M_0,g_0,X_0)\) and not on the cover.

\section{Explicit Lower Bounds on the Parameters of Hyperbolic Color Codes}
\label{sec:explicit-lower-bounds}

In this section, we derive explicit lower bounds on the parameters of the hyperbolic color codes constructed above. We first use the systolic bound of Ref.~\cite{kimSystoleLocallySymmetric2020} to obtain an explicit polynomial lower bound on the code distance in Sec.~\ref{sec:lower-bounds-on-the-code-distance}. We then derive a lower bound on the number of logical qubits in Sec.~\ref{sec:lower-bounds-on-the-number-of-logical-qubits} and show that the resulting codes have constant encoding rate when \(D\) is even and \(i=D/2\).

\subsection{Lower bounds on the code distance}
\label{sec:lower-bounds-on-the-code-distance}

As shown in Sec.~\ref{sec:color-code-lower-bound}, the code distance of a color code is bounded below in terms of the systole of the underlying manifold. To obtain an explicit bound, we successively bound the systole in terms of geometric quantities that can be evaluated for the family of hyperbolic manifolds constructed in Sec.~\ref{sec:hyperbolic-manifold}.

As the first step, we bound the systole in terms of the injectivity radius. The following lemma gives a lower bound on the volume of a homologically nontrivial cycle, and hence on the systole, in terms of the injectivity radius.

\begin{lemma}[Systole lower bound from the injectivity radius{~\cite[Theorem 17]{Guth:2014zxf}}]
    \label{lem:systole-lower-bound-injectivity-radius}
    Let \((M^D,hyp)\) be a closed manifold with a hyperbolic metric. Let \(Z^i\subset M\) be a homologically nontrivial \(i\)-cycle with coefficients in \(\mathbb{Z}_2\). Let \(\inj(M)\) be the injectivity radius of \((M,hyp)\). Then the volume of \(Z\) is at least the volume of a ball of radius \(\inj(M)\) in the \(i\)-dimensional hyperbolic space. In particular, if \(i\ge 2\) and \(\inj(M)\ge 1\), then
    \begin{align}
        \vol_i(Z) \ge c_2(i)\exp({(i-1)\inj(M)}).
    \end{align}
\end{lemma}
This lemma is stated as Theorem 17 in Ref.~\cite{Guth:2014zxf}, where it is derived using a result of Ref.~\cite{andersonCompleteMinimalVarieties1982a}.

We next bound the injectivity radius using its relation to the \(1\)-dimensional systole. Let \(M\) be a closed hyperbolic manifold. The injectivity radius \(\inj(M)\) of \(M\) is half its \(1\)-dimensional systole,
\begin{align}
\label{eq:injectivity-radius}
    \inj(M) = \frac{1}{2}\sys_1(M).
\end{align}

For the arithmetic hyperbolic manifolds considered here, the \(1\)-dimensional systole is bounded below logarithmically in the volume of the manifold. Reference~\cite{kimSystoleLocallySymmetric2020} established this bound for principal congruence subgroups associated with prime ideals. However, the covering condition \(N_0\mid N\) required for \(\Gamma_N\subset\Gamma_{N_0}\) (Sec.~\ref{sec:hyperbolic-manifold}) does not guarantee that the ideals \((N)\subset\mathcal O_F\) indexing the family \(\{M_N\}\) are prime. Therefore we extend this result to arbitrary ideals by adapting the argument of Ref.~\cite[Section~5]{murilloSystoleCongruenceCoverings2019}.

\begin{theorem}[Systole lower bound for principal congruence covers{~\cite{kimSystoleLocallySymmetric2020, murilloSystoleCongruenceCoverings2019}}]
    \label{thm:kim-systole}
    Let \(F\) be a totally real number field and \(\mathcal O_F\) its ring of integers.
    Let \(f\) be a quadratic form of signature \((1,D)\) defined over \(F\), chosen so that
    \(\Gamma=SO_f(\mathcal O_F)\) is a \emph{cocompact} arithmetic lattice of \(SO_f\). For any ideal \(I\subset\mathcal O_F\), let
    \begin{align}
        \Gamma(I) = \Gamma\cap \ker(SO_f(\mathcal O_F)\to SO_f(\mathcal O_F/I))
    \end{align}
    be the corresponding principal congruence subgroup, and set
    \(M_I \coloneqq  \Gamma(I)\backslash\mathbb H^D\). Then
    \begin{align}
    \label{eq:kim-systole}
        \sys_1(M_I) \ge \frac{4}{D(D+1)} \log(\vol(M_I)) - c,
    \end{align}
    where \(c\) is independent of \(I\).
\end{theorem}

The proof of this extension is given in Appendix~\ref{app:composite-ideal-systole} following the strategy used in \cite[Prop.~5.1 and 5.2]{murilloSystoleCongruenceCoverings2019} for the \(\mathrm{Spin}(1,D)\) case.

\begin{remark}
    An explicit systole bound of this type was first derived for a specific four dimensional tessellation of arithmetic hyperbolic manifolds by Ref.~\cite{Londe:2019jie}, which established a distance scaling of \(d=\Omega(n^{0.1})\) for the toric-code family considered there. That work also quotes the stronger bound in Ref.~\cite{murilloSystoleCongruenceCoverings2019}, which would yield \(d=\Omega(n^{0.2})\); however, as they note explicitly, Murillo's theorem~\cite[Theorem 6.1]{murilloSystoleCongruenceCoverings2019} is stated for arithmetic subgroups of \(\mathrm{Spin}(1,D)\), the double cover of \(SO(1,D)^\circ\), and the corresponding principal congruence subgroups are not the same as those of \(\Gamma=G(\mathcal O_F)\subset SO(1,D)^\circ\) used in both constructions. We follow the same convention here and use Theorem~\ref{thm:kim-systole}, which applies directly to \(\Gamma\subset SO(1,D)^\circ\) as constructed in Sec.~\ref{sec:hyperbolic-manifold}.
\end{remark}

Combining Lemma~\ref{lem:systole-lower-bound-injectivity-radius} with Theorem~\ref{thm:kim-systole}, we can express the lower bound on the higher-dimensional systoles directly in terms of the volume of \(M_N\). This yields the following polynomial lower bound.

\begin{corollary}[Explicit polynomial systole bound]
\label{cor:systole-explicit}
Let \(2\le i\le D-1\). For the hyperbolic manifolds \(M_N\) defined in Sec.~\ref{sec:hyperbolic-manifold}, there exists a constant \(c_3(i)>0\) such that, for all sufficiently large \(N\),
    \begin{align}
        \sys_i(M_N) \ge c_3(i)\,\vol(M_N)^{\frac{2(i-1)}{D(D+1)}}.
    \end{align}
\end{corollary}

\begin{proof}
    By Lemma~\ref{lem:systole-lower-bound-injectivity-radius}, every homologically nontrivial \(i\)-cycle in \(M_N\) has volume bounded below exponentially in the injectivity radius. Taking the minimum over all such cycles gives
    \begin{align}
        \sys_i(M_N) &\ge c_2(i)\exp\left((i-1)\inj(M_N)\right) \\
        &= c_2(i)\exp\left((i-1)\frac{1}{2}\sys_1(M_N)\right),
    \end{align}
    where we used Eq.~\eqref{eq:injectivity-radius}.

    Applying the lower bound of Eq.~\eqref{eq:kim-systole}, we obtain
    \begin{align}
        &\sys_i(M_N) \notag \\ 
        &\ge c_2(i) \exp\left(\frac{i-1}{2}\left(\frac{4}{D(D+1)} \log(\vol(M_N)) - c\right)\right) \\
        &= c_2(i)\exp\left(-\frac{(i-1)c}{2}\right)\vol(M_N)^{\frac{2(i-1)}{D(D+1)}} \\
        &= c_3(i) \vol(M_N)^{\frac{2(i-1)}{D(D+1)}}.
    \end{align}
    where \(c_3(i)=c_2(i)\exp(-(i-1)c/2)\) is independent of \(N\). 
\end{proof}

Since the covering map \(p_N\colon M_N\to M_0\) is a local isometry and \(X_N\) is the pullback of \(X_0\) along \(p_N\), every \(D\)-simplex of \(X_N\) has the same volume as its image in \(X_0\), and the number of \(D\)-simplices of \(X_N\) equals \(\deg(p_N)\) times the number of \(D\)-simplices of \(X_0\). Since \(\deg(p_N)=\vol(M_N)/\vol(M_0)\), the number of physical qubits of \(CC_i(X_N)\), which equals the number of \(D\)-simplices of \(X_N\), satisfies
\begin{align}
\label{eq:num-of-qubits-and-volume}
    n = \frac{|\Delta_D(X_0)|}{\vol(M_0)}\,\vol(M_N),
\end{align}
so that \(n\) is proportional to \(\vol(M_N)\) with a constant depending only on the fixed base colex \(X_0\). 

Since the number of physical qubits \(n\) is proportional to \(\vol(M_N)\), the systolic bound of Corollary~\ref{cor:systole-explicit} can be expressed directly in terms of \(n\). Combining this with the color-code distance bound established in Sec.~\ref{sec:color-code-lower-bound} yields an explicit polynomial lower bound on the code distance, as stated in the following theorem.

\begin{theorem}[Explicit polynomial lower bound on the color-code distance]
\label{cor:code-distance-explicit}
Let \(2\le i\le D-2\), and let \(X_N\) be the \(D\)-colex on \(M_N\) obtained by pulling back the fixed colex \(X_0=\mathrm{sd}(Y_0)\) on the base manifold \(M_0\) along the covering map \(p_N\colon M_N\to M_0\). Let \(n\) denote the number of physical qubits of \(CC_i(X_N)\). Then the code distance satisfies
\begin{align}
    d_{CC_i(X_N)} \ge
    \min\left(
        c_4(i)\,n^{\frac{2(i-1)}{D(D+1)}},
        c_4(D-i)\,n^{\frac{2(D-i-1)}{D(D+1)}}
    \right).
\end{align}
\end{theorem}

\begin{proof}
    By Corollary~\ref{cor:systole-explicit} and Eq.~\eqref{eq:num-of-qubits-and-volume}, we obtain
    \begin{align}
        c_1(i)\sys_i(M_N,g_N)
        &\ge c_1(i)c_3(i)\vol(M_N)^{\frac{2(i-1)}{D(D+1)}} \\
        &= c_1(i)c_3(i)\left(\frac{\vol(M_0)}{|\Delta_D(X_0)|}n\right)^{\frac{2(i-1)}{D(D+1)}} \\
        &= c_4(i)n^{\frac{2(i-1)}{D(D+1)}},
    \end{align}
    where \(c_4(i)=c_1(i)c_3(i)(\vol(M_0)/|\Delta_D(X_0)|)^{\frac{2(i-1)}{D(D+1)}}\).

    Applying this inequality for both \(i\) and \(D-i\) in Corollary~\ref{cor:color-code-distance-bounded-by-the-systole}, we obtain the code distance of the color code \(CC_i(X_N)\) satisfies
    \begin{align}
        d_{CC_i(X_N)}\ge \min\left(
        c_4(i)\,n^{\frac{2(i-1)}{D(D+1)}},
        c_4(D-i)\,n^{\frac{2(D-i-1)}{D(D+1)}}
        \right).
    \end{align}
\end{proof}

Table~\ref{tab:systole-exponents} lists the numerical values of the exponent \(\gamma(i,D) \coloneqq \min\left(\frac{2(i-1)}{D(D+1)},\,\frac{2(D-i-1)}{D(D+1)}\right)\) appearing in Theorem~\ref{cor:code-distance-explicit}, for every dimension \(4\le D\le10\) and every admissible type \(2\le i\le D-2\). For a fixed dimension \(D\), the exponent is maximized at \(i=D/2\) when \(D\) is even, while the largest value among the cases shown in the table is attained at \(D=4\) and \(i=2\). An analogous polynomial distance bound was previously established for hyperbolic toric codes by Ref.~\cite{Guth:2014zxf}. Our result establishes rigorously that the same type of polynomial distance scaling extends to hyperbolic color codes in arbitrary dimensions.

An explicit lower bound analogous to that for color codes can also be obtained for the distance of hyperbolic toric codes defined on the same manifolds \(M_N\). Polynomial systolic growth with volume, and the resulting polynomial distance bounds for hyperbolic toric codes, were established in Ref.~\cite{Guth:2014zxf} for all \(D\ge4\) and \(2\le i\le D-2\). However, the corresponding scaling exponents were not explicitly calculated there. Combining Proposition~3 of Ref.~\cite{Guth:2014zxf} with the explicit systolic bound of Corollary~\ref{cor:systole-explicit}, we obtain the following theorem.

\begin{theorem}[Explicit polynomial lower bound on the toric-code distance]
\label{thm:toric-code-distance-explicit}
Let \(2\le i\le D-2\), and let \(X_N\) be the triangulation of \(M_N\) obtained by pulling back a fixed finite smooth triangulation \(X_0\) of the base manifold \(M_0\) along the covering map \(p_N\colon M_N\to M_0\).
Let \(n\) denote the number of physical qubits of \(TC_i(X_N)\).
Then, for all sufficiently large \(N\), the code distance satisfies
\begin{align}
    d_{TC_i(X_N)} \ge
    \min\left(
        c_5(i)\,n^{\frac{2(i-1)}{D(D+1)}},
        c_5(D-i)\,n^{\frac{2(D-i-1)}{D(D+1)}}
    \right),
\end{align}
where \(c_5(i)>0\) are independent of \(N\).
\end{theorem}

\begin{proof}
By Proposition~3 of Ref.~\cite{Guth:2014zxf} and Corollary~\ref{cor:systole-explicit}, there exists a constant \(a>0\), independent of \(N\), such that, for all sufficiently large \(N\),
\begin{align}
    d_{TC_i(X_N)} \ge a\min\Big(
        c_3(i)\,\vol(M_N)^{\frac{2(i-1)}{D(D+1)}}, \notag \\
        c_3(D-i)\,\vol(M_N)^{\frac{2(D-i-1)}{D(D+1)}}
    \Big).
\end{align}
As in Eq.~\eqref{eq:num-of-qubits-and-volume}, the number of physical qubits is proportional to the volume. Since the qubits of \(TC_i(X_N)\) are placed on the \(i\)-simplices, we have
\begin{align}
    n = \frac{|\Delta_i(X_0)|}{\vol(M_0)}\,\vol(M_N).
\end{align}
Define
\begin{align}
    A\coloneqq \min_{2\le r\le D-2}
    \frac{\vol(M_0)}{|\Delta_r(X_0)|}>0.
\end{align}
Then \(\vol(M_N)\ge An\) for every \(2\le i\le D-2\).
Substituting this inequality into the preceding systolic bound yields the claimed result, with \(c_5(j)\coloneqq a c_3(j)A^{\frac{2(j-1)}{D(D+1)}}\).
\end{proof}

\begin{table}[t]
\centering
\caption{Numerical values of the exponent \(\gamma(i,D) = \min\bigl(\frac{2(i-1)}{D(D+1)},\,\frac{2(D-i-1)}{D(D+1)}\bigr)\) in Theorem~\ref{cor:code-distance-explicit}, as a function of the dimension \(D\) and the type \(i\). The values are truncated to three decimal places.}
\label{tab:systole-exponents}
\begin{tabular}{c|ccccccc}
\toprule
\(D \backslash i\) & 2 & 3 & 4 & 5 & 6 & 7 & 8 \\
\midrule
4  & 0.100 &       &       &       &       &       &       \\
5  & 0.066 & 0.066 &       &       &       &       &       \\
6  & 0.047 & 0.095 & 0.047 &       &       &       &       \\
7  & 0.035 & 0.071 & 0.071 & 0.035 &       &       &       \\
8  & 0.027 & 0.055 & 0.083 & 0.055 & 0.027 &       &       \\
9  & 0.022 & 0.044 & 0.066 & 0.066 & 0.044 & 0.022 &       \\
10 & 0.018 & 0.036 & 0.054 & 0.072 & 0.054 & 0.036 & 0.018 \\
\bottomrule
\end{tabular}
\end{table}

\subsection{Lower bounds on the number of logical qubits}
\label{sec:lower-bounds-on-the-number-of-logical-qubits}

As noted in Eq.~\eqref{eq:num_of_logical_qubits_color_code} of Sec.~\ref{sec:color_codes}, the number of logical qubits of the color code \(CC_i(X)\) is proportional to the Betti number \(b_i(X)\), with a proportionality factor depending only on \(D\) and \(i\). Therefore, for fixed \(D\) and \(i\), the number of logical qubits has the same asymptotic scaling as \(b_i(X)\). The Betti numbers of arithmetic hyperbolic manifolds exhibit polynomial scaling with the volume of the manifold.

\begin{lemma}[Polynomial lower bound on Betti numbers of arithmetic hyperbolic manifolds~\cite{xueBettiNumbersHyperbolic1992}]
\label{lem:xue-betti-number-bound}
    Let \(\Gamma\) be an arithmetic lattice in \(SO(D,1)\) which arises from a quadratic form. For any torsion-free congruence subgroup \(\Gamma(p)\) and any \(\epsilon>0\), there is a constant \(c_{\epsilon}>0\), such that for all but finitely many ideals \(q\) of \(p\),
    \begin{align}
        b_i(\Gamma(q)\backslash\mathbb{H}^D) \ge c_{\epsilon} \mathrm{Vol}(\Gamma(q)\backslash \mathbb{H}^D)^{\delta(i,D) -\epsilon}.
    \end{align}
    where \(\delta(i,D)=\frac{(D-1)(D-i)}{(D+1)D}\frac{2i}{D-1}\) for \(i=1,\ldots,\left\lfloor\frac{D+1}{2}\right\rfloor\).
\end{lemma}

This lemma is stated only for \(i\le (D+1)/2\). For closed \(D\)-manifolds, however, Poincaré duality~\cite[Theorem 3.30]{hatcherAlgebraicTopology2002} gives \(b_i(M_N) = b_{D-i}(M_N)\) with \(\mathbb{F}_2\) coefficients. Hence the same polynomial lower bound extends to \(i>(D+1)/2\) by replacing \(i\) with \(D-i\). Accordingly, for \(i>(D+1)/2\), we define \(\delta(i,D)=\delta({D-i}, D)\) so that the exponent \(\delta(i,D)\) is defined for all \(1\le i \le D-1\).

Moreover, when \(D\) is even, the \(D/2\)-th Betti number scales linearly with the volume.

\begin{theorem}[Linear growth of the \(D/2\)-th Betti number]
\label{thm:constant-rate}
    Let \(M_N\) be a \(D\)-dimensional closed manifold defined in Sec.~\ref{sec:hyperbolic-manifold}. If \(D=2m\) for some positive integer \(m\), then
    \begin{align}
    b_m(M_N;\mathbb{F}_2)\ge \frac{2}{\vol(S^D)}\vol(M_N) + o(\vol(M_N)).
    \end{align}
\end{theorem}

\begin{proof}
    The proof is given in Appendix~\ref{app:proof-of-constant-rate}.
\end{proof}

Combining the general Betti-number bound of Lemma~\ref{lem:xue-betti-number-bound} with the linear-growth result of Theorem~\ref{thm:constant-rate}, we obtain the scaling of the number of logical qubits for all \(D\) and \(i\). Using again the fact that the number of physical qubits \(n\) is proportional to \(\vol(M_N)\), these bounds can be expressed directly in terms of \(n\), as summarized in the following theorem.

\begin{theorem}[The number of logical qubits of the color codes]
\label{thm:number-of-logical-qubits-color-codes}
Let \(X_N\) be the \(D\)-colex on \(M_N\) obtained by pulling back the fixed colex \(X_0=\mathrm{sd}(Y_0)\) on the base manifold \(M_0\) along the covering map \(p_N\colon M_N\to M_0\). Let \(n\) denote the number of physical qubits of the color code \(CC_i(X_N)\). For every \(\epsilon>0\), the number of logical qubits satisfies
\begin{align}
    k_{CC_i(X_N)} =
    \begin{cases}
        \Theta(n) & \text{if } D\text{ is even and } i=D/2, \\
        \Omega(n^{\delta(i,D)-\epsilon}) & \text{otherwise}.
    \end{cases}
\end{align}
\end{theorem}

\begin{proof}
Suppose first that \(D\) is even and \(i=D/2\). By Theorem~\ref{thm:constant-rate} and Eq.~\eqref{eq:num_of_logical_qubits_color_code}, we have
\begin{align}
    k_{CC_i(X_N)}
    &= \binom{D}{i} b_i(X_N) \\
    &\ge
    \binom{D}{i}
    \frac{2}{\vol(S^D)}
    \vol(M_N)
    + o(\vol(M_N)) \\
    &=
    \binom{D}{i}
    \frac{2}{\vol(S^D)}
    \frac{\vol(M_0)}{|\Delta_D(X_0)|}
    n
    + o(n).
\end{align}
Thus, \(k_{CC_i(X_N)}=\Omega(n)\). Since the number of logical qubits cannot exceed the number of physical qubits, \(k_{CC_i(X_N)}\le n\), and hence
\begin{align}
    k_{CC_i(X_N)}=\Theta(n).
\end{align}

For all other \(D\) and \(i\), Lemma~\ref{lem:xue-betti-number-bound} together with Eq.~\eqref{eq:num_of_logical_qubits_color_code} gives
\begin{align}
    k_{CC_i(X_N)}
    &= \binom{D}{i} b_i(X_N) \\
    &\ge
    \binom{D}{i} c_{\epsilon}
    \left(
        \frac{\vol(M_0)}{|\Delta_D(X_0)|}n
    \right)^{\delta(i,D)-\epsilon}.
\end{align}
Since all factors other than \(n\) are independent of \(N\), it follows that
\begin{align}
    k_{CC_i(X_N)}
    =\Omega\left(n^{\delta(i,D)-\epsilon}\right).
\end{align}
\end{proof}

Table~\ref{tab:xue-exponents} lists the numerical values of the exponent
\(\delta(i,D)=\frac{2i(D-i)}{D(D+1)}\)
governing the polynomial lower bound on the logical qubit count, for every dimension \(4\le D\le10\) and every admissible type \(2\le i\le D-2\). This exponent is defined for all such \((D,i)\), including the case where \(D\) even and \(i=D/2\); however, in that case the polynomial bound \(\Omega(n^{\delta(i,D)-\epsilon})\) of Lemma~\ref{lem:xue-betti-number-bound} is superseded by the stronger, linear bound \(\Theta(n)\) of Theorem~\ref{thm:constant-rate}, so the tabulated value of \(\delta(i,D)\) at \(i=D/2\) does not represent the true asymptotic growth of \(k_{CC_{D/2}(X_N)}\) in that case.

\begin{table}[t]
\centering
\caption{Numerical values of the exponent \(\delta(i,D)=\frac{2i(D-i)}{D(D+1)}\) governing the polynomial lower bound on the number of logical qubits, as a function of the dimension \(D\) and the type \(i\). At \(i=D/2\) for even \(D\) (boldface), this bound is superseded by the linear bound of Theorem~\ref{thm:constant-rate}. The values are truncated to three decimal places.}
\label{tab:xue-exponents}
\begin{tabular}{c|ccccccc}
\toprule
\(D \backslash i\) & 2 & 3 & 4 & 5 & 6 & 7 & 8 \\
\midrule
4  & \textbf{0.400} &                &                &                &                &                &                \\
5  & 0.400          & 0.400          &                &                &                &                &                \\
6  & 0.380          & \textbf{0.428} & 0.380          &                &                &                &                \\
7  & 0.357          & 0.428         & 0.428          & 0.357          &                &                &                \\
8  & 0.333          & 0.416          & \textbf{0.444} & 0.416          & 0.333          &                &                \\
9  & 0.311          & 0.400          & 0.444          & 0.444          & 0.400          & 0.311          &                \\
10 & 0.290          & 0.381          & 0.436          & \textbf{0.454} & 0.436          & 0.381          & 0.290          \\
\bottomrule
\end{tabular}
\end{table}

The bounds on the number of logical qubits established for our hyperbolic color codes also apply to toric codes defined on the same manifolds. In particular, the corresponding type-\(D/2\) toric codes have constant rate in every even dimension \(D\), extending the result established in Ref.~\cite{Guth:2014zxf} for \(D=4\) and \(i=2\). The argument in that reference uses the Euler characteristic in four dimensions to obtain a lower bound on the second Betti number. In higher dimensions, the Euler characteristic involves additional Betti numbers, so the same argument does not directly yield a lower bound on the relevant Betti number.

To transfer our bounds to toric codes, we use Eq.~\eqref{eq:logical-qubits-color-codes-and-toric-codes}, which shows that the numbers of logical qubits of color and toric codes defined on the same manifold differ only by the constant factor \(\binom{D}{i}\). Moreover, for both code families, the number of physical qubits is proportional to the covering degree, since their underlying complexes are obtained by pulling back fixed complexes on the base manifold. Thus, the bounds established in Theorem~\ref{thm:number-of-logical-qubits-color-codes} also hold for the corresponding toric codes when expressed in terms of their respective numbers of physical qubits. We therefore obtain the following theorem.

\begin{theorem}[The number of logical qubits of the toric codes]
\label{thm:num-of-logical-toric-code}
    Let \(\mathcal L_N\) be the cellulation on \(M_N\) obtained by pulling back a fixed cellulation \(\mathcal L_0\) on the base manifold \(M_0\) along the covering map \(p_N:M_N\to M_0\). Let \(n\) denote the number of physical qubits of the toric code \(TC_i(\mathcal L_N)\). For every \(\epsilon>0\), the number of logical qubits satisfies
    \begin{align}
        k_{TC_i(\mathcal L_N)} = \begin{cases}
            \Theta(n) & \text{if } D \text{ is even and } i=D/2, \\
        \Omega(n^{\delta(i,D)-\epsilon}) & \text{otherwise}.
        \end{cases}
    \end{align}
\end{theorem}

\begin{proof}
The toric code \(TC_i(\mathcal L_N)\) is defined on the same manifold \(M_N\) as the corresponding color code considered in Theorem~\ref{thm:number-of-logical-qubits-color-codes}. Therefore, by Eq.~\eqref{eq:logical-qubits-color-codes-and-toric-codes}, their numbers of logical qubits differ only by the constant factor \(\binom{D}{i}\). Hence the scaling established in Theorem~\ref{thm:number-of-logical-qubits-color-codes} carries over directly to \(TC_i(\mathcal L_N)\), which proves the claim.
\end{proof}

The constant-rate case of Theorem~\ref{thm:num-of-logical-toric-code}, namely \(k=\Theta(n)\) for even \(D\) and \(i=D/2\), was also independently established by Alex Essery~\cite{Essery:2026private}.

Both the color codes and the corresponding toric codes achieve constant rate for \(i=D/2\) in every even dimension \(D\), while in all other cases considered above the number of logical qubits still grows at least polynomially with the number of physical qubits. Combined with the distance bounds derived in the previous section, these results indicate that, for both code families, the most favorable asymptotic parameters are obtained for even \(D\) with \(i=D/2\), and with \(D\) chosen as small as possible.

\section{Conclusion and Outlook}

In conclusion, we have constructed color codes on general manifolds admitting a colex triangulation, and applied this construction to the family of hyperbolic manifolds introduced in Ref.~\cite{Guth:2014zxf} to obtain, for every dimension \(D\geq4\), a hyperbolic color code whose distance grows polynomially in the number of physical qubits \(n\). For even \(D\) and \(i=D/2\), the resulting code has constant rate, while in every other case the number of logical qubits still grows polynomially in \(n\). We have moreover determined the exponent for polynomial lower bounds of the code distance and the number of logical qubits explicitly as a function of \(D\) and \(i\), and compiled these exponents into a catalog of parameters for hyperbolic color codes as shown in Tables~\ref{tab:systole-exponents} and \ref{tab:xue-exponents}.

Looking ahead, several directions remain open. It would be of interest to construct small, explicit examples of our hyperbolic color codes; explicit hyperbolic toric codes have already been obtained in four dimensions~\cite{Londe:2019jie, Breuckmann:2020jyn}, and an analogous construction for the color-code case would be a natural next step. 

A further open problem is to develop a polynomial-time decoder for hyperbolic color codes that simultaneously admits a provable finite threshold and a polynomially growing decoding radius. Existing results for hyperbolic toric codes achieve only part of this goal. An almost-linear-time local decoder with a provable decoding radius of \(O(\log n)\) has been constructed~\cite{hastingsDecodingHyperbolicSpaces2014}. An efficient local decoder with a finite threshold under stochastic noise is also known, although its guarantee against arbitrary errors remains \(O(\log n)\)~\cite{Londe:2019jie}. In contrast, a generalized union-find decoder running in polynomial time provably corrects arbitrary errors of weight up to \(A n^\alpha\), for some \(A,\alpha>0\), but a finite threshold has not been established~\cite{Delfosse:2021rqp}. Such a threshold has been proved for the union-find decoder on the surface code under a circuit-level local stochastic error model~\cite{Yoshida:2026vox}. It therefore remains open to construct a polynomial-time decoder that combines a polynomially growing decoding radius with a provable finite threshold and corresponding exponential suppression in a power of the block size.

Finally, it would be desirable to endow our codes with transversal non-Clifford gates, building on existing constructions of transversal \(T\) gates~\cite{Bombin:2015tpp, Kubica:2014jue}, transversal \(CS\) gates~\cite{Brown:2024kmk}, and transversal \(CCZ\) gates~\cite{Kubica:2015mta, Bombin:2018jad} for color codes. Along similar lines, Ref.~\cite{Zhu:2023xfg} constructed high-rate three-dimensional hyperbolic color codes that admit constant-depth logical \(CCZ\) gates determined by the triple-intersection structure of the underlying manifold. Combined with the nonvanishing rate and polynomial distance established here, such transversal gates could provide a basis for exploring constant-overhead magic-state distillation using qLDPC codes. This direction is motivated by the open question posed in Ref.~\cite{Wills:2024wid} concerning the optimal parameters of qLDPC codes supporting transversally implementable non-Clifford gates.

\begin{acknowledgements}
We acknowledge Tom Scruby and Alex Essery for discussion. This work was supported by JST Moonshot R\&D Grant No. JPMJMS256J, JST PRESTO Grant No. JPMJPR23FC, JST CREST Grant No. JPMJCR25I5,  and Faculty Research Funding from Google Quantum AI.
\end{acknowledgements}

\bibliography{main}

\appendix

\section{Proof of Theorem~\ref{thm:kim-systole}}
\label{app:composite-ideal-systole}

Ref.~\cite{kimSystoleLocallySymmetric2020} proves the systolic bound of Theorem~\ref{thm:kim-systole} only for prime ideals \(I\subset\mathcal O_F\); this restriction enters solely through the group-order bound \(|SO_f(\mathcal O_F/I)|\le N(I)^{D(D+1)/2}\), which in Ref.~\cite{kimSystoleLocallySymmetric2020} is quoted directly from the classification of finite groups of Lie type over a field \(\mathcal O_F/I\). We remove this restriction, following the strategy used in Ref.~\cite[Prop.~5.1--5.2]{murilloSystoleCongruenceCoverings2019} for the \(\mathrm{Spin}(1,D)\) case.

Throughout this appendix, we fix the following notation. Without loss of generality, \(f\) is diagonal in a suitable basis,
\begin{align}
    f = a_1x_1^2 - a_2x_2^2 - \cdots - a_{D+1}x_{D+1}^2, \quad a_i\in\mathcal O_F\setminus\{0\},
\end{align}
so that \(SO_f(R)=\{A\in\mathrm{Mat}_{D+1}(R) \mid A^\top J A = J,\ \det A =1\}\) for \(J=\mathrm{diag}(a_1,-a_2,\dots,-a_{D+1})\) and any commutative \(\mathcal O_F\)-algebra \(R\). We write
\begin{align}
    D_f \coloneqq  \prod_{i=1}^{D+1} a_i \in \mathcal O_F
\end{align}
for the discriminant of \(f\), and let \(S\) denote the (finite) set of prime ideals of \(\mathcal O_F\) dividing \(2D_f\).

The following base case is already established in Sec.~1.1 of \cite{kimSystoleLocallySymmetric2020}; the bound there is stated as \(N(\mathfrak p)^{m(2m-1)}\) for \(D=2m-1\) and \(N(\mathfrak p)^{m(2m+1)}\) for \(D=2m\), both of which equal \(N(\mathfrak p)^{D(D+1)/2}\).

\begin{lemma}[Order bound for \(SO_f\) over a residue field]
    \label{lem:base-case}
    For any prime ideal \(\mathfrak{p}\subset \mathcal O_F\) with no prime factor in \(S\),
    \begin{align}
        |SO_f(\mathcal O_F/\mathfrak{p})| \le N(\mathfrak{p})^{D(D+1)/2}.
    \end{align}
\end{lemma}

Lemma~\ref{lem:base-case} bounds \(|SO_f(\mathcal O_F/\mathfrak p^r)|\) only for \(r=1\). To induct on \(r\), we bound the successive layers of the filtration \(SO_f(\mathcal O_F/\mathfrak p^{r+1})\twoheadrightarrow SO_f(\mathcal O_F/\mathfrak p^{r})\), \(r\ge1\), by computing the kernel of each reduction map directly.

\begin{lemma}[Order of the kernel of the prime-power reduction]
    \label{lem:prime-power-layer}
    For any prime ideal \(\mathfrak p\notin S\) and any \(r\ge1\),
    \begin{align}
        \big|\ker\big(SO_f(\mathcal O_F/\mathfrak p^{r+1})\to SO_f(\mathcal O_F/\mathfrak p^{r})\big)\big|
        = N(\mathfrak p)^{D(D+1)/2}.
    \end{align}
\end{lemma}

\begin{proof}
    Write \(R_r\coloneqq \mathcal O_F/\mathfrak p^r\). An element of the kernel is represented by \(A\in SO_f(R_{r+1})\) whose image in \(SO_f(R_r)\) is the identity, i.e., \(A\equiv I\pmod{\mathfrak p^r}\); we may thus write
    \begin{align}
        A = I + T, \quad T\in \mathrm{Mat}_{D+1}(\mathfrak p^r/\mathfrak p^{r+1}),
    \end{align}
    and the kernel is in bijection with the set of such \(T\) for which \(A\) satisfies the defining equations of \(SO_f\) over \(R_{r+1}\). Since \(\mathfrak p\notin S\), each \(a_i\) and \(2\) are units in the local ring \(R_{r+1}\). In particular \(J=\mathrm{diag}(a_1,-a_2,\dots,-a_{D+1})\) is invertible over \(R_{r+1}\).

    As \(r\ge1\), every entry of \(T\) lies in \(\mathfrak p^r/\mathfrak p^{r+1}\), so every entry of \(T^\top JT\) lies in \(\mathfrak p^{2r}\subseteq\mathfrak p^{r+1}\), i.e., \(T^\top JT\equiv0\). Hence
    \begin{align}
        A^\top JA &= J + T^\top J + JT + T^\top JT \notag \\
        &\equiv J + T^\top J + JT \pmod{\mathfrak p^{r+1}},
    \end{align}
    so \(A^\top JA=J\) is equivalent to the linear condition
    \begin{align}
        T^\top J + JT \equiv 0 \pmod{\mathfrak p^{r+1}}. \label{eq:linearized}
    \end{align}
    Likewise \(\det(I+T)\equiv 1+\mathrm{tr}(T)\pmod{\mathfrak p^{r+1}}\), since every other term in the expansion of the determinant involves a product of at least two entries of \(T\) and hence lies in \(\mathfrak p^{2r}\subseteq\mathfrak p^{r+1}\); so \(\det A=1\) is equivalent to
    \(\mathrm{tr}(T)\equiv0\).

    The diagonal entries of Eq.\eqref{eq:linearized} give \(2J_{ii}T_{ii}\equiv0\) for each \(i\); since \(\mathfrak p\nmid 2D_f\), both \(2\) and \(J_{ii}\in\{a_1,-a_2,\dots,-a_{D+1}\}\) are units in \(R_{r+1}\), so \(2J_{ii}\) is a unit and this forces \(T_{ii}\equiv0\) for all \(i\), so \(\mathrm{tr}(T)\equiv0\) holds automatically and the determinant condition imposes no further
    constraint. The kernel is thus exactly
    \begin{align}
        \{I+T \mid T\in\mathrm{Mat}_{D+1}(\mathfrak p^r/\mathfrak p^{r+1}),\ T^\top J+JT\equiv0\}.
    \end{align}

    Since \(J^\top=J\), condition \eqref{eq:linearized} is equivalent to \((JT)^\top=T^\top J=-JT\), i.e., \(JT\) is skew-symmetric. As \(J\) is invertible over \(R_{r+1}\), \(T\mapsto JT\) is a bijection carrying the solutions of Eq.\eqref{eq:linearized} onto the set of skew-symmetric matrices with entries in \(\mathfrak p^r/\mathfrak p^{r+1}\). The space of skew-symmetric \((D+1)\times(D+1)\) matrices over \(\mathcal O_F/\mathfrak p\cong\mathbb F_{N(\mathfrak p)}\) has dimension \(\binom{D+1}{2}=D(D+1)/2\); since \(\mathfrak p^r/\mathfrak p^{r+1}\) is a one-dimensional \(\mathbb F_{N(\mathfrak p)}\)-vector space, the corresponding set of skew-symmetric matrices with entries in \(\mathfrak p^r/\mathfrak p^{r+1}\) has cardinality \(N(\mathfrak p)^{D(D+1)/2}\). Combining the above bijections,
    \begin{align}
        \big|\ker\big(SO_f(\mathcal O_F/\mathfrak p^{r+1})\to SO_f(\mathcal O_F/\mathfrak p^{r})\big)\big|
        = N(\mathfrak p)^{D(D+1)/2}.
    \end{align}
\end{proof}

Combining Lemmas~\ref{lem:base-case} and~\ref{lem:prime-power-layer} via the Chinese Remainder Theorem now yields the order bound for \(SO_f(\mathcal O_F/I)\) at an arbitrary (not necessarily prime) ideal \(I\).

\begin{proposition}[Order bound for \(SO_f(\mathcal O_F/I)\) at composite ideals~{\cite[Prop.~5.2]{murilloSystoleCongruenceCoverings2019}}]
    \label{prop:composite-order-bound}
    For any ideal \(I\subset\mathcal O_F\) with no prime factor in \(S\),
    \begin{align}
        |SO_f(\mathcal O_F/I)| \le N(I)^{D(D+1)/2}.
    \end{align}
\end{proposition}
\begin{proof}
    Write \(I=\prod_i \mathfrak p_i^{r_i}\). By the Chinese Remainder Theorem, \(SO_f(\mathcal O_F/I)\cong\prod_i SO_f(\mathcal O_F/\mathfrak p_i^{r_i})\). For each \(i\), Lemma~\ref{lem:base-case} gives \(|SO_f(\mathcal O_F/\mathfrak p_i)|\le N(\mathfrak p_i)^{D(D+1)/2}\), and Lemma~\ref{lem:prime-power-layer} gives, by induction on \(r_i\), \(|SO_f(\mathcal O_F/\mathfrak p_i^{r_i})|\le N(\mathfrak p_i)^{r_i D(D+1)/2}\). Multiplicativity of \(N\) over coprime ideals gives the claim.
\end{proof}

\begin{lemma}[Index bound via the order of \(SO_f(\mathcal O_F/I)\)~{\cite[Prop.~5.1]{murilloSystoleCongruenceCoverings2019}}]
    \label{lem:index-bound}
    For any ideal \(I\subset\mathcal O_F\),
    \begin{align}
        [\Gamma:\Gamma(I)] \le |SO_f(\mathcal O_F/I)|.
    \end{align}
\end{lemma}

\begin{proof}
    Reduction modulo \(I\) gives a homomorphism \(\rho\colon\Gamma=SO_f(\mathcal O_F)\to SO_f(\mathcal O_F/I)\) with \(\ker\rho=\Gamma(I)\). By the first isomorphism theorem, \(\Gamma/\Gamma(I)\cong\operatorname{im}\rho\subseteq SO_f(\mathcal O_F/I)\), so \([\Gamma:\Gamma(I)]\le|SO_f(\mathcal O_F/I)|\).
\end{proof}

\begin{proof}[Proof of Theorem~\ref{thm:kim-systole}]
    As in the derivation of the theorem in Ref.\cite[Section 1.1]{kimSystoleLocallySymmetric2020}, the trace–displacement estimate for elements of \(\Gamma(I)\subset SO(1,D)\) depends on the ideal \(I\) only through \(N(I)\), and does not require \(I\) to be prime. It therefore gives, for arbitrary \(I\subset\mathcal O_F\),
    \begin{align}
        \sys_1(M_I) \ge 2\log(N(I)) - c,
    \end{align}
    where \(c\) is independent of \(I\).

    On the other hand, Lemma~\ref{lem:index-bound} together with Proposition~\ref{prop:composite-order-bound} gives
    \begin{align}
        [\Gamma:\Gamma(I)] \le |SO_f(\mathcal O_F/I)| \le N(I)^{D(D+1)/2},
    \end{align}
    i.e., \(N(I) \ge [\Gamma:\Gamma(I)]^{2/D(D+1)}\). Since \(\vol(M_I)=[\Gamma:\Gamma(I)]\vol(M)\), combining this with the displacement estimate above yields
    \begin{align}
        \sys_1(M_I) &\ge 2\log(N(I)) - c \notag\\
        &\ge \frac{4}{D(D+1)}\log\bigl([\Gamma:\Gamma(I)]\bigr) - c \notag\\
        &= \frac{4}{D(D+1)}\log\!\left(\frac{\vol(M_I)}{\vol(M)}\right) - c \notag\\
        &= \frac{4}{D(D+1)}\log(\vol(M_I)) - d,
    \end{align}
    where \(d \coloneqq  c + \frac{4}{D(D+1)}\log(\vol(M))\) is independent of \(I\subset\mathcal O_F\).
\end{proof}

\section{Proof of Theorem~\ref{thm:constant-rate}}
\label{app:proof-of-constant-rate}

In this appendix we prove Theorem~\ref{thm:constant-rate}, using the result of Ref.~\cite{abertGrowth$L^2$invariantsSequences2017a} to control the growth of Betti numbers of \(M_N\).

\begin{lemma}[Benjamini--Schramm convergence of \(M_N\) to \(\mathbb{H}^D\)]
\label{lem:BS}
The family \(\{M_N\}_{N_0\mid N}\) Benjamini--Schramm (BS) converges to \(\mathbb H^{D}\) as \(N\to\infty\), i.e., \(\vol((M_N)_{<R})/\vol(M_N)\to0\) for every \(R>0\), where
\((M)_{<R}\) denotes the set of points of injectivity radius less than \(R\).
\end{lemma}

\begin{proof}
By Theorem~\ref{thm:kim-systole} and \(\inj(M_N)=\tfrac12\sys_1(M_N)\), there are constants \(c_1,c_2>0\) independent of \(N\) such that \(\inj(M_N)\ge c_1\log(\vol(M_N))-c_2\), and \(\vol(M_N)\) increases monotonically with \(N\). Hence, given \(R>0\), we have \((M_N)_{<R}=\emptyset\) for all sufficiently large \(N\).
\end{proof}

\begin{theorem}[Convergence of normalized Betti numbers under BS-convergence~{\cite[Thm.~7.13]{abertGrowth$L^2$invariantsSequences2017a}}]
\label{thm:abbgnrs}
Let \((M_n)\) be a sequence of closed hyperbolic \(D\)-manifolds that BS-converges to \(\mathbb H^{D}\). Then for every \(0\le k\le D\),
\begin{align}
 \lim_{n\to\infty}\frac{b_k(M_n;\mathbb Q)}{\vol(M_n)}=\beta^{(2)}_k(\mathbb H^{D}),
\end{align}
where \(\beta^{(2)}_k(\mathbb H^{D})=\chi(S^{D})/\vol(S^{D})\) if \(k=D/2\) and \(\beta^{(2)}_k(\mathbb H^{D})=0\) otherwise.
\end{theorem}

\begin{lemma}[Universal coefficients]\label{lem:F2}
For any closed manifold \(M\), \(\dim_{\mathbb F_2}H_k(M;\mathbb F_2)\ge b_k(M;\mathbb Q)\).
\end{lemma}

\begin{proof}
\(H_k(M;\mathbb F_2)\cong (H_k(M;\mathbb Z)\otimes\mathbb F_2)\oplus \operatorname{Tor}(H_{k-1}(M;\mathbb Z),\mathbb F_2)\), and the first summand already has dimension at least \(\operatorname{rank}H_k(M;\mathbb Z)=b_k(M;\mathbb Q)\).
\end{proof}

\begin{proof}[Proof of Theorem~\ref{thm:constant-rate}]
Let \(D=2m\). By Lemma~\ref{lem:BS} and Theorem~\ref{thm:abbgnrs} with \(k=m\), \(b_m(M_N;\mathbb Q)/\vol(M_N)\to \chi(S^{2m})/\vol(S^{2m})=2/\vol(S^{2m})\). Lemma~\ref{lem:F2} then gives, for every \(0<c<2/\vol(S^{2m})\) and all sufficiently large \(N\), \(\dim_{\mathbb F_2}H_m(M_N;\mathbb F_2)\ge c\,\vol(M_N)\).
\end{proof}

\end{document}